\documentclass[a4paper,11pt]{article}
\usepackage{amsmath,amsfonts}
\newtheorem{theorem}{Theorem}[section]
\newtheorem{definition}[theorem]{Definition} 
\newtheorem{lemma}[theorem]{Lemma} 
\newtheorem{proposition}[theorem]{Proposition} 
\newtheorem{corollary}[theorem]{Corollary} 
 
\newtheorem{remark}[theorem]{Remark}

\newenvironment{proof}[1][Proof]{{\sc #1.} }{\hfill $\Box$}
\newcommand{\wW}{\widehat W}
\newcommand{\dist}{\mathop{\mathrm{dist}}\nolimits}
\newcommand{\expect}{\mathbb{E}}
\newcommand{\smallsetminus}{\backslash}
\newcommand{\R}{\mathbb{R}}
\newcommand{\YY}{\mathbb{Y}}
\newcommand{\XX}{\mathbb{X}}
\newcommand{\ZZ}{\mathbb{Z}}
\newcommand {\C}  {\ensuremath {\mathcal{C}}}
\newcommand {\A}  {\ensuremath {\mathcal{A}}}
\newcommand {\D}  {\ensuremath {\mathcal{D}}}
\newcommand {\V}  {\ensuremath {\mathcal{V}}}
\newcommand {\G}  {\ensuremath {\mathcal{G}}}
\newcommand {\T}  {\ensuremath {\mathcal{T}}}
\newcommand {\U}  {\ensuremath {\mathcal{U}}}

\newcommand{\E}{\mathop{\mathrm{{\cal E}}}\nolimits}
\newcommand{\F}{\mathop{\mathrm{{\cal F}}}\nolimits}
\newcommand{\J}{\mathop{\mathrm{{\cal J}}}\nolimits}
\newcommand{\K}{\mathop{\mathrm{{\cal K}}}\nolimits}
\newcommand{\W}{\mathop{\mathrm{{\cal W}}}\nolimits}
\newcommand{\WW}{\mathop{\mathrm{\widehat{\cal W}}}\nolimits}
\newcommand{\p}{\mathop{\mathrm{{\cal P}}}\nolimits}

\newcommand{\X}{\mathcal{X}}
\newcommand {\Z} {\mathbb{Z}}
\newcommand {\norm} [2] [] {\ensuremath{ \left\Vert  #2  \right\Vert_{#1} } }

\newcommand{\nuxT}{\nu_T^{\mathbf{x}}}

\newcommand{\WxT}{\mathcal{W}_T^{\mathbf{x}}}

\newcommand{\eps}{\varepsilon}

\newcommand{\diam}{\mathop{\mathrm{diam}}\nolimits}
\newcommand{\Spec}{\mathop{\mathrm{Spec}}\nolimits}
\newcommand {\Ga} {\ensuremath{\Gamma}}
\newcommand {\ga} {\ensuremath{\gamma}}
\newcommand {\rh} {\ensuremath{\varrho}}
\newcommand{\e}{\mathop{\mathrm{{\varepsilon}}}\nolimits}
\newcommand{\cov}{\mathop{\mathrm{cov}}\nolimits}

\makeatletter
\@addtoreset{equation}{section}
\makeatother

\newcounter {constant}
\newenvironment{constant}{\refstepcounter{constant} }{}

\newcommand{\const}[1]{C_{\ref{const:#1}}}
\newcommand{\defconst}[1]{\begin{constant}\label{const:#1}\const{#1}\end{constant}}

\begin{document}

\title
{\Large \bf Gibbs measures on Brownian currents}
\vspace{1.5cm}
\author{
\small Massimiliano Gubinelli\\[0.1cm]
 {\small\it    Laboratoire de Math\'ematiques,  
 Universit\'e de Paris-Sud }    \\[-0.7ex]
  {\small\it  B\^atiment 425,   F-91405 Orsay Cedex, France}      \\[-0.7ex]
 {\small  {\tt massimiliano.gubinelli@math.u-psud.fr}   }\\[0.5cm]
\small J\'ozsef L\H{o}rinczi\\[0.1cm]
{\it \small Zentrum Mathematik, Technische Universit\"at M\"unchen} \\[-0.7ex]
{\it \small Boltzmannstr. 3, 85747 Garching bei M\"unchen, Germany} \\[-0.7ex]
{\small {\tt  lorinczi@ma.tum.de}} \\[-0.7ex]}
\date{}
\maketitle
\vspace{2cm}
\begin{abstract}
\noindent
Motivated by applications to quantum field theory we consider Gibbs
measures for which the reference measure is Wiener measure and the
interaction is given by a double stochastic integral and a pinning 
external potential. In order properly to characterize these measures 
through DLR equations, we are led to lift Wiener measure and other 
objects to a space of configurations where the basic 
observables are not only the position of the particle at all times 
but also the work done by test vector fields. We prove existence and 
basic properties of such Gibbs measures in the small coupling regime 
by means of cluster expansion.
\end{abstract}
\newpage
\section{Motivations and outline}
Gibbs measures are familiar objects in various areas of applied
probability. Originally they have been devised in the framework of 
lattice spin systems to describe thermodynamic equilibrium states. 
Obtaining these probability measures involves two basic steps. 
First a family of Gibbs measures for a finite number of random
variables is constructed as a modification of a reference measure 
usually describing the independent random field. The density is the
exponential of an additive functional dependent on an interaction 
function. Then one takes a weak limit of these measures by increasing 
the number of variables to infinity. This procedure copes with the
fact that in usual model systems the interaction diverges in this
limit, thus the limit measure cannot be directly defined. 

The context in which the class of Gibbs probability measures proved 
to be useful has substantially widened over the years, reaching the 
realm of Brownian motion. 
One natural way leading to Gibbs measures on path space is the 
application of the Feynman-Kac formula and its extensions. Suppose 
$H$ is a self-adjoint operator acting on a Hilbert space $\cal H$, and 
let $f \in {\cal H}$. Then for a variety of specific choices of $H$ 
an equality of the type 
\begin{equation}
(e^{-tH}f)(x) = \int_{\R^d} dy \int f(X_t) d\mu_{[0,t]}(X|x,y)
\label{fk}
\end{equation}
can be derived, where $\R \ni t\mapsto X_t \in \R^d$ is Brownian 
motion and
\begin{equation}
d\mu_{[0,t]}(X|x,y) = e^{-U_t(X|x,y)} d\W^{x,y}_{[0,t]}(X).
\label{gibb}
\end{equation}
Here $\W^{x,y}_{[0,t]}$ is Brownian bridge over the bounded time 
interval $[0,t]$, starting at $X_0 = x$ and ending at $X_t = y$, and 
$U_t$ is a functional of Brownian paths derived from $H$. After 
normalizing to 1, the measure $\mu_{[0,t]}(X|x,y)$ can be viewn as a 
Gibbs measure on path space for ``finite volume'' $[0,t]$,  
``interaction'' $U$, ``boundary condition'' $x,y$, and ``reference 
measure'' $\W^{x,y}$.  

In this paper we are interested in the question whether an extension 
in a suitable sense over the full time line $\R$ of $\mu_{[0,t]}$ 
exists. The details of the proof strongly depend on the choice of $U$, 
which for our purposes here will be specified below. These choices are 
motivated by particular applications covered by the following classes
of models, of which we talk only in this introduction. 

\medskip
\vspace{0.1cm}
\noindent
I. Densities dependent on the local time of Brownian motion
\begin{enumerate}
\item[(1)]
{\it $P(\phi)_1$-process (It\^o diffusion)} \quad 
This includes the familiar case of Schr\"odinger operators $H = 
(-1/2)\Delta + V$, ${\cal H} = C_0^\infty(\R^d)$, with a potential 
$V(x)$ that can be chosen fairly generally (Kato-class). The 
well-known result is 
$$
U_t(X) = \int_0^t V(X_s) ds.
$$
$\mu$ describes in this case the path measure of the Markov process 
given by 
$$
dX_t = dB_t + (\nabla\log\Psi)(X_t)dt,
$$
where $B_t$ is $\R^d$-valued Brownian motion and $\Psi$ is the
eigenfunction of $H$ lying at the bottom of its spectrum (ground
state). It\^o diffusions can be defined also on infinite dimensional 
spaces related to other SPDE. For further details on these Gibbs 
measures for bounded intervals and their extensions to $\R$ we refer 
to \cite{Sim75,Sim82,BL}.

\item[(2)] 
{\it Nelson's model} \quad 
This is a scalar quantum field model describing the interaction 
of an electrically charged spinless particle with a boson field. In
this case $H$ is written as the sum of the free particle Hamiltonian
$(-1/2)\Delta + V$, the free field Hamiltonian $\int |k|a(k)^*a(k)dk$
with the usual boson creation and annihilation operators $a^*$ and $a$,
and the interaction Hamiltonian $\int(\widehat\rho(k)/{\sqrt{2|k|}})
(e^{ik\cdot x}a(k) + e^{-ik\cdot x}a^*(k))dk$, with charge
distribution function $\rho$. Moreover, ${\cal H} = L^2(\R^d,dx)
\otimes \F$, where $\F$ is Fock space, and a Feynman-Kac-type formula 
as (\ref{fk}) above can be obtained by mapping $\cal H$ into a space 
of continuous functions through a joint use of the so called ground 
state transform and Wiener-It\^o isomorphism yielding
$$
U_t(X) = \int_0^t V(X_s) ds + \int_0^t \int_0^t W^\rho(X_s-X_r,s-r) 
dsdr,
$$
with 
$$
W^\rho(x,s) = -\frac{1}{4}\int_{\R^d} \frac{|\widehat \rho(k)|^2}{|k|} 
e^{-i k\cdot x - |k||s|} dk.
$$
Extending this Gibbs measure from $[0,t]$ to $\R$ 
is in this case of special interest since it allows a direct 
expression of the ground state of $H$ in terms of its Radon-Nikodym 
derivative with respect to an underlying product measure, which makes 
a rigorous derivation and proof of ground state properties possible. 
For the case of translation invariant models (involving $V \equiv 0$) 
there are only few results available. For details see 
\cite{LM,LMS1,LMS2,BHLMS}. 

\item[(3)]
{\it Polaron and bipolaron models}\quad
The polaron is a ``dressed'' electron (i.e., embedded into an energy
cloud) interacting with a phonon field (i.e., quantum particles
carrying the vibrational energy of an ionic crystal). In this case we
have similar operators acting on the same Hilbert space as above 
except that the dispersion relation $|k|$ is replaced by 1 in the
free field Hamiltonian, and $\widehat\rho/\sqrt{2|k|}$ in the 
interaction term by $1/|k|$. This leads to the same $U_t$ as in the 
case of Nelson's model with 
$$
W^{\rm\tiny{pol}}(x,s) = -\frac{1}{4|x|}\,e^{-|s|}.
$$
The bipolaron differs by the fact that it consists of two dressed
electrons coupled to the same phonon field, which are repelling each
other by Coulomb interaction. In this case 
$$
U_t(X) = \alpha^2 \int_0^t\int_0^t {\cal E}^W(X_s,Y_r,s-r) dsdr
- g\int_0^t \frac{ds}{|X_s-Y_s|},
$$
where 
$$
{\cal E}^W(X_s,Y_r,u) = W(X_s-X_r,u) + 2W(X_s-Y_r,u) + W(Y_s-Y_r,u)
$$
with $W = W^{\rm\tiny{pol}}$, $\alpha < 0$ being the polaron-phonon
coupling paramater and $g>0$ the strength of the Coulomb repulsion 
between the two polarons. Moreover, in this case the reference measure
is a product of two independent Wiener measures. For literature see
\cite{DV,S,LM06}.

\item[(4)]
{\it Intersection local time (weakly self-avoiding polymer)}\quad
Formally, the densities are given by
$$ 
U_t(X) = \int_0^t \int_0^t \delta(X_s-X_r) ds dr, 
$$ 
meant to describe a polymer model with short-range ``soft-core'' 
interaction encouraging to avoid self-intersections. For $d=2, 3$ 
see \cite{MR853758,MR667754,MR609228,MR573702,MR1240717}. In 
\cite{Symanzik:1969lr} it was proved that in $d=2$ the model can be 
rigorously defined after an additive renormalization and the so 
obtained measure is absolutely continuous with respect to Wiener measure 
in 2 dimensions. In $d=3$ the singularity of the energy $U_t$ is more 
serious but an additive renormalization still suffices; Westwater proved 
existence of the Gibbs measure which, however, in this case is not 
absolutely continuous with respect to Wiener measure. Other works 
include \cite{MR729794,MR1329112,MR889758,MR942042}.

\end{enumerate}

\vspace{0.1cm}
\noindent
II. Stochastic currents
\begin{enumerate}
\item[(1)]
{\it Nelson's model in point charge limit and Pauli-Fierz model} \quad
The point charge limit of Nelson's model corresponds to the case of 
replacing $\rho$ above with a delta-function, while the Pauli-Fierz 
model is obtained by replacing the scalar boson field with a quantized
Maxwell (vector) field. The function in the density in both cases 
formally becomes 
\begin{equation}
U_t(X) = \int_0^t V(X_s) ds + \int_0^t \int_0^t W(X_s-X_r,s-r) 
dX_s \cdot dX_r,
\label{doubleito}
\end{equation}
with a $W=W^\delta$ and $W=W^{\rm PF}$ we do not write explicitly 
down here. The difference from the case above is that instead of 
double Riemann we have to deal with double It\^o integrals. A
substantial difficulty to solve in this case is the proper definition 
of the double integrals in the first place, which will be done below. 
For the specific model applications see \cite{GL} in which we perform 
an ultraviolet renormalization of Nelson's Hamiltonian by using path
measures whose densities are of the above form, and \cite{H,HL2} for the
Pauli-Fierz model. The similar point charge limit in the latter model 
is an open problem.

\item[(2)]
{\it Turbulent fluids} \quad
In fluid dynamics the understanding prevails that fully developed 
turbulence should be described by a suitable measure over
divergence-free velocity fields $u(x)$. One way of modelling it starts 
from the assumption that the vorticity field $\nabla \wedge u(x)$ is 
concentrated along Brownian curves $X_t \in \mathbb{R}^3$. Under the 
Eulerian incompressible flow, the kinetic energy $(1/2) \int u(x)^2 
dx$ is conserved. The formal expression of the total energy is
\begin{equation}\label{3.11}
U_t(X)= \int_0^t \int_0^t \frac{1}{|X_t - X_s|} \; dX_t \cdot dX_s.
\end{equation}
In order to have $e^{-U_t(X)}$ as a well-defined random variable at
all, \cite{Ffil} imposed the condition that the Coulomb potential is 
mollified so that the fluid has finite static energy. For details we
refer to~\cite{statvort,FGub,Nua,evolution,Ascona}. 
\end{enumerate}

\vspace{0.1cm}
\noindent
{III. Processes with jumps} \quad 
\begin{enumerate}
\item[]
Applications include Hamiltonians with spin. Since spin is a discrete 
variable, this case goes beyond stochastic integration with respect to
Brownian motion alone. The spin paths $\sigma_t$ are described via a 
Poisson process $N_t$, and we obtain
$$
U_t(X) = \int_0^t V(X_s) ds + \int_0^t a(X_t) \circ dX_t - 
\int_0^t S(X_t,\sigma_t) dt +\int_0^{t^+} \Phi(X_t,-\sigma_t) dN_t.
$$
For more details and explicit formulas of the terms above see
\cite{Jona,HL1}.
\end{enumerate}

\medskip
In this paper we address the purely mathematical problem of existence
and characterization of Gibbs measures with densities of the type 
(\ref{doubleito}). An accompanying paper \cite{GL} takes this further 
to an application to quantum field theory. 

The problem can be formulated in more generality by adding a term 
to $U_t$ in (\ref{doubleito}) taking into account the dependence on 
outside paths of paths run within bounded time intervals. Choose
without loss a bounded interval in the form $[-T,T] \subset \R$ and 
set up Wiener measure $\W_T$ on it. Consider the energy functions 
corresponding to $[-T,T]$ as follows: 
\begin{eqnarray*}
&&
V_T(X) := \int_{-T}^T  V(X_s) ds \qquad \mbox{for external 
potential $V$}\\
\bigskip
&&
W_T(X) :=  \int_{-T}^T \int_{-T}^T W(X_t-X_s,t-s) dX_t \cdot dX_s 
\qquad \mbox{internal energy} \\ 
\medskip\bigskip
&&
W_T(X|Y):= W^+_T(X|Y) + W^{-}_T(X|Y) \qquad 
\mbox{interaction energy}
\end{eqnarray*}
where the interaction energies come from a pair interaction 
potential $W$, and $W^+_T(X|Y) := \int_T^\infty dY_t \int_{-T}^T 
dX_s W(X_s-Y_t,t-s)$ is a term accounting for the interaction of 
paths inside $[-T,T]$ with paths in $[T,\infty)$, and $W^-_T(X|Y) 
:= \int_{-\infty}^{-T} dY_t \int_{-T}^T dX_s W(X_s-Y_t,t-s)$ for 
external paths running in $(-\infty,-T]$. As in (\ref{gibb}), 
these energies give rise to the Gibbs measure for $[-T,T]$
\begin{equation}
\label{ea:gibbs1}
d\mu_T(X) := \frac{e^{-V_T(X)-\lambda W_T(X)}}{Z_T}\, d\W_T(X)
\end{equation}
with a parameter $\lambda \in \R$ tuning the strength of the pair 
potential, and $Z_T$ the normalizing factor turning it into a
probability measure. The object of interest are the accumulation 
points $\mu$ of the family of measures $\{\mu_{T}\}_{T>0}$ in the 
topology of local weak convergence. In the DLR 
(Dobrushin-Lanford-Ruelle) approach, these limit points can be 
characterized by a property of consistency with respect to a 
prescribed family of probability kernels (specification) providing
the local conditional probabilities of the limit random fields. These 
kernels are given by 
\begin{equation}
\label{eq:specc}
\mu_T(dX|Y) := \frac{e^{-V_T(X)-\lambda W_T(X)-
\lambda W_T(X|Y)}}{Z_T(Y)} \, d\W^{Y_{-T},Y_{T}}_T(X)
\end{equation}
with external (or boundary) path $Y$. The possible accumulation 
points of the family $\{\mu_T\}_T$ should then satisfy the DLR 
equations
$
\int \mu_T(A|Y) \mu(dY) = \mu(A)
$ 
for all cylinder sets $A$ in the sub-$\sigma$-field generated by 
projections to $[-T,T]$ for all $T>0$.

In contrast to the case of interactions depending on the local time 
of the process $X$ (given by double Riemann integrals listed above) 
our case encounters two difficulties.
\begin{itemize}
\item[(1)]  
The expressions of $W_T(X)$ and $W_T(X|Y)$ are only formal: 
The double stochastic integrals are not well defined since the 
integrands are neither forward nor backward adapted with respect 
to the semimartingale $X$ (Brownian bridge under $\W^{x,y}_T$).
\item[(2)]  
The specification (\ref{eq:specc}) must be defined pathwise for each 
$Y$, however, in general the only information we have on the boundary 
path $Y$ is that it is a continuous path with a Brownian-like
regularity. This is insufficient for defining the line integrals with 
respect to $dY$ appearing in $W_T(X|Y)$.
\end{itemize}

\noindent
Of these two difficulties the first requires a minor amendment, 
while the second is far more serious and will urge us to introduce 
the novel setting of \emph{Brownian currents} in which we can make
sense of a formulation of the DLR equations. Roughly speaking, we 
provide each sample path with sufficient information in order to 
determine the work made by test vector fields. Specifically, we 
consider the family of random variables $\{X_t\}_t$ jointly with 
the random variables 
$
 C_{st}^X( \varphi) := \int_s^t \varphi(u,X_u) dX_u  
$
(\emph{stochastic currents})
and show that in this augmented sample space the specification can 
pathwise be defined. (This problem has some analogy with difficulties 
at defining specifications for unbounded spin systems for which too 
one has to select a subset of admissible boundary conditions so that 
the interaction energy makes sense. Our approach here goes further in 
that we do not only have to control the growth of the boundary paths 
but must also provide a priori the value of the line integrals against 
a sufficiently large set of test vector fields $\varphi$.) To perform 
this ``lifting'' procedure we use techniques of rough paths theory. 

\medskip
Here is an outline of the paper. In Chapter 2 
we present our results on rough paths that will be put at use 
subsequently. In Chapter 3 we introduce the framework of stochastic 
currents, in particular Brownian currents. In Chapter 4 we define
Gibbs measures on these currents for bounded intervals of the line. In 
Chapter 5 we turn to proving that these Gibbs measures can be extended 
over the whole line provided the interaction term is weakly coupled. 
This requirement is needed since, \emph{faute de mieux}, we use cluster 
expansion in order to construct weak limits of Gibbs measures for 
bounded intervals. The version of cluster expansion we develop here is 
different from the conventional ones since the configurations are 
segments of Brownian paths rather than spins of compact or real-valued 
state space, the interactions depend on double stochastic integrals, 
the reference measure is not a product measure, and the measure we want 
to construct is non-Markovian. However, it has the same spirit of usual 
cluster expansions as it splits off into a part of hands-on analysis of 
energy bounds and a part of combinatorics. Beside a proof of 
existence, cluster expansion allows us to study also uniqueness, typical
path behaviour and mixing properties of Gibbs measures, which will be
done in Chapter 6.

\label{sec:gibbs}

\section{Rough paths}
\subsection{Definition of double stochastic integrals by rough paths}
\label{sec:rough-paths}

It\^o's theory of integration gives a meaning to line integrals of the 
form 
$
\int_{-T}^T \varphi(t,X_t) dX_t,
$
with $X$ a semimartingale (with respect to its natural forward
filtration $\mathcal{F}$) and $\varphi(t,\xi)$ some function. The way
of obtaining it is by taking limits of Riemann sums
$
\sum_\alpha \varphi(\tau_\alpha,X_{\tau_\alpha}) 
(X_{\tau_{\alpha+1}}-X_{\tau_\alpha})
$
over a family of partitions $\{\tau_\alpha\}_\alpha$ of $[-T,T]$ with 
mesh decreasing to zero. Such limits are known to exist whenever 
$\varphi(t,X_t)$ is an $\mathcal{F}$-adapted functional of $X$, i.e.,
if it only depends on $\F_t$ for any $t\in[-T,T]$ and has some 
integrability properties. Then 
$
\int_{-T}^T \varphi(t,X_t) dX_t
$
is defined as a random variable on the same probability space on which 
$X$ is defined. In our applications $X$ is the coordinate process on the
Borel space $\X$, and the stochastic integral defines a Borel map from
$\X$ to $\R$ which is defined on a set of full $\W$ measure.
It is important to remember the crucial fact that the full-measure set 
in general depends on the integrand: different integrands may give 
different full-measure sets.

In this section we analyze the regularity of such integrals by using
the theory or rough paths as developed in 
\cite{lyons-98,lyons-qian-02,MR2091358,feyel}. This theory allows to 
properly understand stochastic integrals from a purely analytic 
perspective. To gain this freedom we need to enlarge the sample space 
over which the measures are defined, however, this will turn out to 
make no harm in our applications.

Rough paths theory has been devised by T.~Lyons~\cite{lyons-98} to
give a meaning to line integrals of the form
\begin{equation}
  \label{eq:basic-rp-integral}
\int_0^T \varphi(X_t) dX_t,  
\end{equation}
in case of $X$ being an irregular function of the parameter. A typical
case of interest is when $X$ is chosen to be H\"older continuous with 
small exponent, such as $\gamma\in(0,1)$. 
The natural approach to defining such integrals is that of taking
Riemann approximations over a finite partition of $[0,T]$ and proving 
that the sequence converges as the mesh of the partition goes to zero. 
When general integrals in the form $\int_0^T Y_s dX_s$ are considered
with bounded $Y$, this can only hold if $X$ is a process of bounded 
variation leading to the familiar Stieltjes integral. When $Y$ is 
H\"older continuous with exponent $\rho$, Young's work \cite{young-36} 
makes sure that a sufficient condition for the convergence of the 
Riemann sums is that $\gamma+\rho > 1$. The Young integral is useful 
in studying, for example, fractional Brownian motion of Hurst index 
$H>1/2$, which is a stochastic process with H\"older continuous paths 
of any exponent $\gamma < H$. In this case (\ref{eq:basic-rp-integral}) 
can be defined for any function $\varphi$ which is at least $C^1$.

Integrals of the form (\ref{eq:basic-rp-integral}) when $X$ is a
sample of Brownian motion are not within the reach of Young's theory 
and indeed it is not difficult to see that different definitions of 
the Riemann approximations can lead to different limits (or even not 
converge at all). This difficulty lies at the basis of the existence 
of more than one type of stochastic integral over Brownian motion, 
two well known possibilities being the integrals developed by It\^o 
and Stratonovich. Lyons's approach to the problem was proving that 
under reasonable conditions the Riemann sums for $\int_0^T \varphi(X_t) 
dX_t$ can be modified by adding an extra term in order to make them
converge. This compensation is not unique but in many relevant cases 
it can be performed in such a way that the so obtained integral is an 
extension of the classical line integral and Young integral. That is,
whenever $X$ is regular enough, the modified integral coincides with 
the Young one or Riemann-Lebesgue integral if moreover $X$ is almost 
surely is differentiable.

We consider only the case when $X$ is H\"older continuous with $\gamma 
> 1/3$ since this includes the case of Brownian motion and all our 
present applications. (The theory for more general $\gamma$ would be 
far more cumbersome.) We use the notation $X_{st} = X_t-X_s$, and make 
the following basic

\begin{definition}
For each bounded $I \subset \R$ we call the couple $(X,\mathbb{X})$  
a \emph{2-step rough path}, where
\begin{enumerate}
\item[(1)]
$X \in C^\gamma(I,\R^d)$;

\item[(2)]
there exists a function $\mathbb{X} : I\times I \to \R^{d}\times \R^d$ 
satisfying the \emph{multiplicative property}
  \begin{equation}
    \label{eq:mult}
  \mathbb{X}^{ij}_{st} - \mathbb{X}^{ij}_{su} - \mathbb{X}^{ij}_{ut} = 
  X^i_{su} X^j_{ut}   
  \end{equation}
for all $s \le u \le t \in I$ and any $i,j =1,\dots,d$;

\item[(3)]
there is $C \in [0,\infty)$ such that
\begin{equation}
  \label{eq:rough-hold-area}
\|\mathbb{X}_{st}\|_{2\gamma} := \sup_{t \neq s} 
\frac{|\mathbb{X}_{st}|}{|t-s|^{2\gamma}} \le C.
\end{equation}
\end{enumerate}
\end{definition}
\noindent
We use the similar norm $\|{X_{st}}\|_\gamma$ correspondingly given
by the right hand side of (\ref{eq:rough-hold-area}).

\begin{remark}
\rm{
Note that Assumption 
(2) is non-trivial only if $d>1$. When $d=1$ we can take
$
\mathbb{X}_{st} = (X_{st})^2/2,
$
and it can easily be seen that both (\ref{eq:mult}) and the
H\"older-like bound~(\ref{eq:rough-hold-area}) are then satisfied. 
}
\end{remark}

For our purposes below we need a slight extension since we want 
to be able to integrate functions explicitly dependent on the time
parameter, such as $\int_0^T \varphi(u,X_u) dX_u$. One possibility 
is to consider the couple $(u,X)$ as a new path and construct the 
associated rough path. This construction, however, has the disadvantage 
of requiring rather much regularity of the function $\varphi$. A more 
effective approach is to treat the line-integral as a Young integral 
with respect to the $u$ dependence of the integrand as soon as the 
map $u\mapsto \varphi(u,\xi)$ has H\"older exponent greater than 
$1-\gamma$ (as we shall see below).

For the (time-dependent) test vector fields in $C(\R \times \R^d; \R^d)$
 we use the norm
$$
\|\varphi\|_{\rho,2,s,t} = \sup_x \left[ \sup_{u\in[s,t]}
\left(\max_{k=0,1,2} |\nabla^k \varphi(u,x)|\right) +
 \sup_{u,v \in[s,t]} \max_{k=0,1}
\frac{|\nabla^k \varphi(u,x)-\nabla^k \varphi(v,x)|}{|u-v|^\rho} 
\right]
$$
with $\rho > 0$ and the convention $\nabla^0 \varphi(t,x) = 
\varphi(t,x)$. 
Then our basic result on the step-2 rough path $(X,\mathbb{X})$  is

\begin{theorem}
\label{th:rough-main}
Let $\varphi \in C(\R \times \R^d,\R^d)$ be such that it is $C^2$ 
with respect to its second variable and H\"older continuous with
exponent $\rho$ with respect to its first variable, such that 
$\rho+\gamma > 1$. Then the sums
$$
\sum_\alpha \left(\varphi_i(\tau_\alpha,X_{\tau_\alpha}) 
X^i_{\tau_{\alpha+1} \tau_\alpha}+\nabla_j 
\varphi_i(\tau_\alpha,X_{\tau_\alpha}) 
\mathbb{X}_{ \tau_\alpha \tau_{\alpha+1}}^{ij} \right), 
\quad i,j=1,...,d
$$
converge as the mesh of the partition $\{\tau_\alpha\}_{\alpha}$ of
$[0,T]$ goes to zero, and defines the integral
$
\int_0^T \varphi(u,X_u) dX_u 
$ 
Moreover for any $T\le 1$ we have the bound
$$
\left|\int_0^T \varphi(u,X_u) dX_u \right| \le C T^\gamma  
\|\varphi\|_{\rho,2,0,T} (1+\|X\|_\gamma+\|\mathbb{X}^2\|_{2\gamma})^3.
$$
\end{theorem}
\begin{proof}
We prove convergence over the dyadic partition of $[0,1]$, then 
convergence of a general partition follows then by the arguments 
developed in~\cite{MR2091358}. Let $\tau^{(n)}_{\alpha}=T \alpha/2^n$ 
for $\alpha = 0,\dots,2^n$, and let
$$
S_n = \sum_{\alpha=0}^{2^n-1} 
\left(\varphi_i(\tau^{(n)}_\alpha,X_{\tau^{(n)}_\alpha}) 
X^i_{\tau^{(n)}_{\alpha+1}
\tau^{(n)}_\alpha}+\nabla_j 
\varphi_i(\tau^{(n)}_\alpha,X_{\tau^{(n)}_\alpha}) 
\mathbb{X}_{\tau^{(n)}_{\alpha+1}
  \tau^{(n)}_\alpha}^{ij}\right). 
$$
Note that
\begin{equation}
\label{eq:rough-split}
 S_{n}-S_{n-1} = \sum_{\alpha=0}^{2^n-1} 
(A_\alpha^1 + A_\alpha^2  +
 A_\alpha^3 + A_\alpha^4),
\end{equation}
where
\begin{eqnarray*}
&& 
A^1_\alpha  =
\left[\varphi_i(\tau^{(n)}_{2\alpha+1},X_{\tau^{(n)}_{2\alpha+2}})-
\varphi_i(\tau^{(n)}_{2\alpha+1},X_{\tau^{(n)}_{2\alpha+1}})-\nabla_j
\varphi_i(\tau^{(n)}_{2\alpha+1},X_{\tau^{(n)}_{2\alpha+1}}) 
X^j_{\tau^{(n)}_{2\alpha+2} \tau^{(n)}_{2\alpha+1}} \right] 
X^i_{\tau^{(n)}_{2\alpha+1} \tau^{(n)}_{2\alpha}}, \\
&& 
A^2_\alpha  =
\left[\nabla_j\varphi_i(\tau^{(n)}_{2\alpha+1},X_{\tau^{(n)}_{2\alpha+2}})- 
\nabla_j \varphi_i(\tau^{(n)}_{2\alpha+1},X_{\tau^{(n)}_{2\alpha+1}})\right] 
\mathbb{X}_{\tau^{(n)}_{2\alpha+1} \tau^{(n)}_{2\alpha}}^{ij},\\
&&
A^3_\alpha  = 
\left[\varphi_i(\tau^{(n)}_{2\alpha+2},X_{\tau^{(n)}_{2\alpha+2}})-
\varphi_i(\tau^{(n)}_{2\alpha+1},X_{\tau^{(n)}_{2\alpha+2}}) \right] 
X^i_{\tau^{(n)}_{2\alpha+1}\tau^{(n)}_{2\alpha}},\\
&& 
A^4_\alpha  = 
(\nabla_j\varphi_i(\tau^{(n)}_{2\alpha+2},X_{\tau^{(n)}_{2\alpha+2}})-
\nabla_j \varphi_i(\tau^{(n)}_{2\alpha+1},X_{\tau^{(n)}_{2\alpha+2}})) 
\mathbb{X}_{\tau^{(n)}_{2\alpha+1}\tau^{(n)}_{2\alpha}}^{ij}.
\end{eqnarray*}
By using the fact that $\varphi\in C^2$, we have the bounds 
$
|A^1_\alpha | \le T^{3\gamma}\|\varphi\|_{\rho,2,0,T} \|X\|^3_\gamma
2^{-3\gamma n}
$
and
$
|A^2_\alpha| \le  T^{3\gamma} \|\varphi\|_{\rho,2,0,T} \|X\|_\gamma
\|\mathbb{X}\|_{2\gamma} 2^{-3\gamma n}
$
on the first two terms.
For the last two we have
$
|A^3_\alpha| \le  T^{\rho+\gamma} \|\varphi\|_{\rho,2,0,T}
\|X\|_\gamma 2^{-n(\rho+\gamma)}
$
and 
$
| A^4_\alpha | \le  T^{\rho+2\gamma} \|\varphi\|_{\rho,2,0,T}
\|\mathbb{X}\|_{2\gamma} 2^{-n(\rho+2\gamma)}
$.
Thus
$$
|S_{n}-S_{n-1}| \; \le \; C (T^{3\gamma} +
T^{\rho+\gamma})\|\varphi\|_{\rho,2,0,T} (2^{(1-3\gamma) n}
+ 2^{(1-\gamma-\rho)n}) 
$$
where $C$ is a function dependent only on $X,\mathbb{X}$. Write
$S_n$ as the telescopic sum 
$
S_n = S_0 + \sum_{k=1}^n(S_k-S_{k-1})
$.
Then, since
$$
\sum_{k=1}^n |S_k-S_{k-1}| \le C (T^{3\gamma}+T^{\rho+\gamma}) 
\|\varphi\|_{\rho,2,0,T} \sum_{k=1}^n (2^{(1-3\gamma) n}+2^{(1-\gamma-\rho) n})
$$
converges geometrically whenever $\gamma > 1/3$ and $\gamma+\rho>1$,
we obtain absolute convergence of the sequence $S_n$. The bound on the 
integral thus follows.
\end{proof}

\medskip
As this theorem shows, by using the additional information provided by 
the path $\mathbb{X}$ we are able to make Riemann sums converge. The 
following consequence is immediate.
\begin{corollary}
The integral defined in Theorem \ref{th:rough-main} has
$
\mathbb{X}_{st}^{ij} = \int_{s}^t X^i_{su} dX^j_u
$.  
\end{corollary}
\begin{proof}
We have
$$
\int_{s}^t X^{i}_{us} dX^j_u = \int_{s}^t X^i_{u} dX^j_u - X^i_s 
(X^j_u-X^j_s),
$$ 
as it is easily seen by using the definition. That is, the rough 
integral has the same linearity property as Riemann integral and
behaves the same way with respect to the integration of constants. 
Moreover, the integral in the right hand side corresponds to the 
function $\varphi_k(\xi) = \xi^i \delta_{kj}$ and by using 
(\ref{eq:rough-split}) we have $S_{n}=S_{n-1}$ on the dyadic partition 
above. Hence $S_n = S_0$ and 
$$
\int_{s}^t X^i_{u} dX^j_u = S_0 = \sum_k \varphi_k(X_s) X^k_{st} +
\sum_{k,m} \nabla_m \varphi_k(X_s) \mathbb{X}_{st}^{km} = X_s^i
X^j_{st} + \mathbb{X}_{st}^{ij},
$$
which completes the proof.
\end{proof}

\begin{remark}
\rm{
The function $\mathbb{X}$ can be identified as giving the value of a 
\emph{twice iterated integral} over $X$,
\begin{equation}
  \label{eq:rough-ident}
\mathbb{X}_{st}^{ij} = \int_s^t \left(\int_s^u dX^i_v\right) dX^j_u. 
\end{equation}
Thus Theorem \ref{th:rough-main} can be alternatively interpreted as 
saying that the knowledge of the twice iterated integral (in addition 
with some H\"older continuity) is sufficient for determining the value 
of the integral $\int \varphi(X_t)dX_t$ for arbitrary $C^2$ function.  
}
\end{remark}

Provided $\gamma > 1/2$ and (\ref{eq:rough-ident}) holds, whenever the
right hand side is defined by using the Young integral, the integral 
defined in Theorem \ref{th:rough-main} coincides with the Young
integral. This is easy to see due to the estimate 
$$
|\mathbb{X}_{st}^{ij}| \le C (\|X\|_{\infty}+\|X\|_\gamma)^2 
|t-s|^{2\gamma}
$$
for the Young integral. Since $2\gamma > 1 $, the sums
$
\sum_\alpha \nabla_j \varphi_i(X_{\tau_\alpha}) 
\mathbb{X}_{\tau_{\alpha+1} \tau_\alpha}^{ij} 
$
vanish as the mesh goes to zero, so the modified and Riemann sums 
converge to the same limit.
Moreover, the rough integral is continuous in the natural topology
associated to the step-2 rough path $(X,\mathbb{X})$, i.e., we have
\begin{corollary}
Let $(X^n,\mathbb{X}^n)$ be a sequence of step-2 rough paths such 
that
$
\|X^n - X\|_\gamma + \|\mathbb{X}^n-\mathbb{X}\|_{2\gamma} \to 0.
$  
Then
$
\int \varphi(X^n) dX^n \to \int \varphi(X) dX,
$
for all $\varphi \in C^2$.
\end{corollary}

Suppose $X^n$ is a sequence of regular (say, piecewise linear)
approximations of the H\"older continuous path $X$. Then by putting 
$\mathbb{X}^n = \int \int dX^n\otimes  dX^n $, where the integrals 
are Riemann integrals, a sufficient condition for the convergence of 
the approximate integrals $\int \varphi(X^n) dX^n$ to $\int \varphi(X) 
dX$ is that the sequence $\mathbb{X}^n$ converges to $\mathbb{X}$ in 
the topology induced by $\|\cdot\|_{2\gamma}$.

\subsection{Brownian motion as rough path}

Let now $X$ be a sample path of Brownian motion. By Kolmogorov's Lemma 
the process $t \mapsto X_t$ has a version that is H\"older continuous 
with exponent $\gamma \in(1/3,1/2)$ (actually $\gamma$ can be taken 
arbitrarily close to $1/2$). In the following we will use such a 
version without each time mentioning explicitly, i.e., consider a 
subspace $\X_\gamma \subset C(\R,\R^d)$ such that every $X \in 
\X_\gamma$ is H\"older continuous with exponent $\gamma$.
 
To apply the above results to $X$ we need a choice for $\mathbb{X}$. 
This candidate is not unique, different choices will lead to different 
integrals over $X$. For instance, in order to construct a possible 
$\mathbb{X}$ we can start by setting
$$
(\mathbb{X}_{\text{It\^o}})^{ij}_{st}=\int_s^t \int_s^u dX^i_v dX^j_u,
$$  
where the double integral is understood in It\^o sense. In this way we
obtain a family of random variables
$\{(\mathbb{X}_{\text{It\^o}})^{ij}_{st}: i,j=1,\dots,d; \; t,s
\in[0,T]\}$ satisfying the multiplicative property
\begin{equation*}
(\mathbb{X}_{\text{It\^o}})^{ij}_{st} - 
(\mathbb{X}_{\text{It\^o}})^{ij}_{ut} -
(\mathbb{X}_{\text{It\^o}})^{ij}_{su} 
= X^i_{ut} X^j_{su}   
\end{equation*}
almost surely for any fixed $t,s,u\in[0,T]$. The next step is to show 
that this family has a version for which
\begin{equation}
  \label{eq:rough-holder-area}
\|\mathbb{X}^{ij}_{\text{It\^o}}\|_{2\gamma} < \infty  \quad \mbox{a.s.}
\end{equation}
In order to prove (\ref{eq:rough-holder-area}), we use the following
lemma obtained as an extension of a result of Garsia-Rodemich-Rumsey 
in \cite{MR2091358}.
\begin{lemma}
\label{lemma:besov}
For any $\theta > 0$ and $p \ge 1$ there exists a constant $C$ such
that for any $R \in C([0,T]^2,{\cal B})$, where $({\cal B},|\cdot|)$ 
is a Banach space, we have 
\begin{equation}
\label{eq:generalboundxx}
\|R\|_{\theta} \le C \left(U_{\theta+2/p,p}(R)+V_\theta(R)\right),
\end{equation}
with
\begin{equation*}
 U_{\theta,p}(R) = \left[ \int_{[0,T]^2}
 \left(\frac{|R_{t s}|}{|t-s|^\theta}\right)^p dt ds \right]^{1/p}
\end{equation*}
and
\begin{equation*}
 V_\theta(R) = \inf_{\theta_1 \in(0,\theta)} \sup_{t \neq u \neq s} 
\frac{|R_{st}-R_{ut}-R_{su}|}{|t-u|^{\theta_1}|u-s|^{\theta-\theta_1}}. 
\end{equation*}
\end{lemma}

\begin{corollary}
We have $\|\mathbb{X}_{{\text{It\^o}}}\|_{2\gamma} < \infty$ $\W$-almost 
surely.  
\end{corollary}
\begin{proof}
Consider $V_{2\gamma}(\mathbb{X}_{\text{It\^o}})$. By using the 
multiplicative property we have
\begin{eqnarray*}
V_{2\gamma}(\mathbb{X}_{\text{It\^o}}) 
&\le&  
\inf_{\theta_1 \in (0,2\gamma)} \sup_{t \neq u \neq s} 
\frac{|X_{tu}||X_{us}|}{|t-u|^{\theta_1}|u-s|^{2\gamma-\theta_1}}\\ 
&\le&  
\inf_{\theta_1 \in(0,2\gamma)} \sup_{t \neq u} \frac{|X_{tu}|}
{|t-u|^{\theta_1}}\sup_{u \neq s}
\frac{|X_{us}|}{|u-s|^{2\gamma-\theta_1}}\\
&\le&  
\left(\sup_{t \neq u} \frac{|X_{tu}|}{|t-u|^{\gamma}}\right)^2\\
& = &
\|X\|_\gamma^2.
\end{eqnarray*}
Moreover,
\begin{eqnarray*}
\expect [U_{2\gamma,p}(\mathbb{X}_{\text{It\^o}})^p] 
&=& 
\expect  \left[\int_{[0,T]^2}\left(\frac{|(\mathbb{X}_{\text{It\^o}})_{t s}|}
{|t-s|^{2\gamma}}\right)^p dt ds \right]\\
&=&    
\int_{[0,T]^2}\frac{ \expect [|(\mathbb{X}_{\text{It\^o}})_{t s}|^p]}
{|t-s|^{2\gamma p}} dt ds, 
\end{eqnarray*}
with expectation with respect to Wiener measure. An application of the 
Burkholder-Davis-Gundy inequality~\cite{Yor} allows to estimate the 
$p$-moment of the double stochastic integral as
\begin{eqnarray*}
\expect [|(\mathbb{X}_{\text{It\^o}})_{t s}|^p] 
& \le&
c_p \expect \left[\int_s^t |X_{us}|^2 ds\right]^{p/2} \\
&\le&
c_p |t-s|^{p/2-1} \expect \left[\int_s^t |X_{us}|^p ds\right] \\ 
& \le &
c_p' |t-s|^{p/2-1}\int_s^t |u-s|^{p/2} ds \\
&\le &
c_p'' |t-s|^{p},   
\end{eqnarray*}
for all $p > 1$ and some $c_p, c_p', c_p'' > 0$. Thus
\begin{eqnarray*}
\expect [U_{2\gamma,p}(\mathbb{X}_{\text{It\^o}})^p] 
&\le& 
c_p'' \int_{[0,T]^2} \frac{ 1}{|t-s|^{(2\gamma-1) p}} dt ds < \infty,
\end{eqnarray*}
for any $\gamma < 1/2$, by choosing $p$ large enough ($p > 1/(1-2\gamma)$).
\end{proof}
 
This last result shows that there exists a version of the stochastic process
$(X,\mathbb{X}_{\text{It\^o}})$ that is a step-2 rough path; from now on we 
denote by $(X,\mathbb{X}_{\text{It\^o}})$ this particular version. Then 
integrals can be defined by applying Theorem \ref{th:rough-main}. We
call such an integral \emph{rough integral} over $(X,\mathbb{X}_{\text{It\^o}})$.

The relationship between the rough integral and the It\^o integral is made
clear by 

\begin{lemma}
The rough integral over the couple $(X,\mathbb{X}_{\text{It\^o}})$
coincides with the It\^o integral for any $\varphi \in C^2$.  
\end{lemma}
\begin{proof}
By It\^o theory the sums
$
\sum_\alpha \varphi_i(X_{\tau_\alpha}) X^i_{\tau_{\alpha+1} \tau_\alpha}
$ 
converge in probability to the It\^o integral $\int_0^T \varphi(X_u) dX_u$.
Hence it suffices to show that the sums
$
\sum_\alpha \nabla_j \varphi_i(X_{\tau_\alpha}) \mathbb{X}_{\tau_{\alpha+1}
  \tau_\alpha}^{ij} 
$
converge to zero in $L^2$ sense as then it follows that the two
integrals almost surely coincide. A simple computation shows that
\begin{equation}
\expect\left[ \sum_\alpha \nabla_j \varphi_i(X_{\tau_\alpha}) 
\mathbb{X}_{\tau_{\alpha+1} \tau_\alpha}^{ij} \right]^2 
=  
\sum_\alpha \expect \left[  \nabla_j \varphi_i(X_{\tau_\alpha}) 
\mathbb{X}_{\tau_{\alpha+1} \tau_\alpha}^{ij} \right]^2,
\nonumber
\end{equation}
since the cross terms are all zero in the mean by independence of the
increments of the Brownian motion. Hence 
\begin{equation*}
  \begin{split}
& \expect\left[ \sum_\alpha \nabla_j \varphi_i(X_{\tau_\alpha}) 
  \mathbb{X}_{\tau_{\alpha+1} \tau_\alpha}^{ij} \right]^2
    \le \|\varphi\|_1 \sum_\alpha \expect \left[ \mathbb{X}_{\tau_{\alpha+1}
  \tau_\alpha}^{ij} \right]^2
    \le C \|\varphi\|_1 \sum_\alpha |\tau_{\alpha+1}
 - \tau_\alpha|^{2}
  \end{split}
\end{equation*}
The last sum vanishes as the mesh of the partition shrinks to zero, thus
the claim follows.
\end{proof}

In this construction choosing the It\^o version for the double integral
$\mathbb{X}$ was arbitrary. Alternatively we could have considered other
definitions, e.g. Stratonovich integral and let
$$
(\mathbb{X}_{\text{Stra}})_{st}^{ij} = \int_s^t \int_s^u \circ dX^i_v
\circ dX^j_u
$$
where $\circ dX$ stands for Stratonovich (or symmetric) integration. By 
the same procedure we find a regular version of $\mathbb{X}_{\text{Stra}}$ 
such that $\|\mathbb{X}_{\text{Stra}}\|_{2\gamma} < \infty$ and can 
construct the rough integral over the couple $(X,\mathbb{X}_{\text{Stra}})$ 
which we denote again $\circ dX$. It is not difficult to prove that for
$\varphi \in C^2$ it coincides with the familiar Stratonovich integral.

The relationship between the rough integrals based on the It\^o and 
Stratonovich constructions follows from the identity
$$
(\mathbb{X}_{\text{Stra}})_{st}^{ij} =
(\mathbb{X}_{\text{It\^o}})_{st}^{ij} + \frac{1}{2} \delta_{ij} (t-s)
$$
between It\^o and Stratonovich stochastic iterated integrals. The
correction is given by the increment of the function $t \mapsto
\delta_{ij} t/2$. Thus the two rough integrals are related by the 
familiar formula
$$
\int_0^T \varphi(X_u) \circ dX_u = \int_0^T \varphi(X_u) dX_u +
\frac{1}{2} \int_0^T \text{div } \varphi(X_u) du.
$$
This is obtained directly from the definitions with the modified
Riemann sums. 

\begin{remark}
\rm{
Due to the multiplicative property the possible choices for 
$\mathbb{X}$ differ only by the increment of a function, i.e., 
if $\mathbb{X}_1$ and $\mathbb{X}_2$ both satisfy the 
multiplicative property with respect to $X$, then
$
(\mathbb{X}_1)_{st}^{ij}- (\mathbb{X}_2)_{st}^{ij} = 
f^{ij}_t-f^{ij}_s
$ 
with a function $f \in C([0,T],\R^d\times \R^d)$.   
}
\end{remark}

Finally, it is not difficult to see the following regularity result.
\begin{corollary}
The map $\varphi \mapsto \int_0^1 \varphi(X_u) dX_u$ is continuous
from $C^2(\R^d,\R^d)$ to $\R$.  
\end{corollary}

\section{Stochastic currents}
\subsection{Lifting Wiener measure to the space of currents}
Let $\X = C(\R, \R^d)$ be path space, i.e., the space of continuous
functions from $\R$ to $\R^d$. The $\sigma$-field $\F$ is generated by 
the coordinate process $\X \ni X \mapsto X_t \in \R^d$. For $I \subset 
\R$ we denote by $\F_I$ the $\sigma$-field over $\X$ generated by the 
evaluations for points in $I$, and write $\F_T$ when $I =
[-T,T]$. Also, we put for a shorthand $I^c = \R \setminus I$. 
For $s \leq t$, the set $\{\mathcal{F}_{st} = 
\sigma(X_u : s \le u \le t)\}_{s \leq t}$ is the forward filtration 
starting at $s \in \R$.

Denote as before by $\W$ Wiener measure defined on $\X$ (instead of 
$C([0,\infty),\R^d)$ as more usual), and write $H_0 = -(1/2)\Delta$. 
For any finite division $t_1 < t_2 < ... < t_n \in \R$  we have
\begin{equation} \label{wiener}
\W(F) = (f_1,e^{-(t_2-t_1)H_0} f_2 \ldots 
e^{-(t_n-t_{n-1})H_0} f_n)_{L^2(\R^d,dx)}
\end{equation}
where $F = f_1(X_{t_1})\cdots f_n(X_{t_n})$. Here it is understood
that $f_2, \ldots, f_{n-1}$ act as multiplication operators for which 
we use the same symbol as for the corresponding functions. The
operator $e^{-tH_0}$ has the integral kernel 
\begin{equation} \label{(2.2)}
\Pi_t(x,y) = \frac{1}{(2 \pi t)^{d/2}} 
\exp \left( - \frac{1}{2t} (x-y)^2 \right), \quad x,y \in \R^d.
\end{equation}
We denote by $\W_I$ the measure $\W$ restricted to $\X_I = C(I,\R^d)$.
Similarly, with given $T>0$, $\xi, \eta \in \R^d$  write 
\begin{equation} \label{(2.3)}
\W^{\xi,\eta}_I(F) = \left( e^{-(-T-t_1)H_0} f_1 e^{-(t_2-t_1)H_0}f_2 
\ldots e^{-(t_n-t_{n-1})H_0} \left( f_n e^{-(T-t_n)H_0}(\cdot,\eta) 
\right) \right) (\xi)
\end{equation}
with $I=[-T,T]$ and let $\WW^{\xi,\eta}_I(F) = \W^{\xi,\eta}_I(F)/ 
\Pi_{2T}(\xi,\eta)$ be the Brownian bridge starting in $\xi$ at $-T$ 
and ending in $\eta$ at $T$. Under $\WW^{x,y}_{[-T,T]}$ the process 
$X_t$ is a Gaussian semimartingale (Brownian bridge) satisfying the SDE
$$
dX_u = - \frac{y-X_u}{T-u} du + dB_u
$$
where $(B_u)_{u\in[-T,T]}$ is a $\WW^{x,y}_{[-T,T]}$-Browian motion. 

\medskip
Next, let $\D$ be a Banach space of functions from $\R \times \R^d$
to $\R^d$ containing the space of smooth functions 
$C^\infty(\R \times \R^d; \R^d)$. Denote with $\|\cdot\|_\D$ the
Banach norm. Let $\D'$ the topological dual of $\D$, $\Delta_I =
\{(s,t)\in I^2 | s\le t\}$, and write $\Delta = \Delta_\R$. 
\begin{definition}
$C \in C(\Delta ; \D')$ is a \emph{stochastic current} if it 
satisfies the following properties:
\begin{enumerate}
\vspace{0.1cm}
\item[(1)]
$
C_{tt}(\varphi) = 0, \;\; C_{su}(\varphi) + C_{ut}(\varphi) = 
C_{st}(\varphi), \;\;
$
for any $s\le u\le t$ and any $\varphi \in \D$;
\vspace{0.1cm}
\item[(2)]
\emph{locality property:} \quad $C_{st}(\varphi)=0 \;$ 
whenever $\;\varphi(u,x)=0$ for all $u\in[s,t]$, $x \in \R^d$.
\end{enumerate}
We denote $\C \subset C(\Delta ; \D')$ the space of stochastic
currents. 
\end{definition}
Set $\Xi = \X \times \C$ endowed with the product topology and 
with the Borel $\sigma$-field (on whose component generated
by $\C$ we consider the topology of uniform convergence on bounded 
intervals). $\Xi$ plays the role of joint path-current configuration 
space. As a measurable space, it is endowed with a family of 
$\sigma$-algebras 
$\{\A_{st}\}_{t>s}$ such that $\A_{st} = \sigma(X_u, C_{uv}(\varphi) 
: u,v \in [t,s], \varphi \in \D )$. Similarly, we can define the 
forward filtration $\{\A_{t}^+\}_t = \{\A_{+\infty,t}\}_t$ and the 
backward filtration $\{\A^-_t\}_t = \{\A_{t,-\infty}\}_t$. The above 
definitions make sense also in the case that the parameter $t$ is 
restricted to a bounded interval $I\subset \R$; in this case we denote 
with $\Xi_I$ the corresponding space. Whenever the limits make sense 
we define $C^{+}_t(\varphi) = \lim_{s\to\infty} C_{ts}(\varphi)$ and 
$C^-_t(\varphi) = \lim_{s\to -\infty}C_{st}(\varphi)$.

\begin{definition}
\label{def:forward-current}
A \emph{forward current} (on $I\subset \R$) is a measure $\eta$ on
$\Xi$ (on $\Xi_I$) such  that the process $X$ is an 
$(\eta,\{\F_t^+\}_t)$-semimartingale and
\begin{equation}
  \label{eq:curr-constr-0}
C_{st}^X(\varphi) = \int_s^t \varphi(u,X_u) dX_u, 
\qquad \eta-\text{a.s.}  
\end{equation}
for any $(s,t)\in \Delta$ (or $\Delta_I$) and any adapted $\varphi \in 
\D$ where on the right hand side we have the standard It\^o integral on 
the semimartingale $X$. When $X$ is Brownian motion, we call $\eta$ 
\emph{(forward) Brownian current}. 
\end{definition}
In order not to multiply terminology, unless confusion may arise we 
will use the term \emph{current} also for the elements of $\D'$, of 
$\Xi$, and for the laws on $\Xi$ without making explicit distinction. 
For unspecified bounded intervals $I$ we use the notation $C_I^X$ for
the associated current with integrator $X$.

Next we want to construct a (non-trivial) measure on $\Xi$ for the
Brownian current. Thus we start from Wiener measure $\W$ (similarly 
we could have worked with the Brownian bridge $\W_I^{x,y}$) and prove 
that there exists a map $F: \X \ni \omega \mapsto F(\omega)\in\C$ such
 that
\begin{equation}
  \label{eq:curr-constr-0-bis}
F(\cdot)_{st}(\varphi) = \int_s^t \varphi(u,X_u) dX_u, 
\qquad \W-\text{a.s.}  
\end{equation}
for any adapted $\varphi \in \D$, with the standard It\^o integral 
at the right hand side. Then a measure $\W^\sharp$ on $\Xi$ can be 
defined as the law of the couple $(X,F)$ under the measure $\W$ and 
it will be a forward current. The existence of a regular version of 
map $F$ is an interesting problem in itself which can be addressed 
by using the techniques developed in~\cite{FGGT, FGR} for what 
concerns the regular dependence on $\varphi$. Unfortunately, the 
topology $\D$ which is implied by such approaches is unsuitable for 
our applications. Here we prefer to use the theory of rough paths 
which will provide the necessary regularity for the $F$ map in a 
more convenient topology.

In Section ~\ref{sec:rough-paths} we developed the basic tools of
rough-path theory that we need in order to lift Wiener measure to 
currents. We do this next.

For any $\alpha > 1$, let $\D_\alpha$ be the completion of the space of 
smooth test vector fields with respect to the norm
$$
\|\varphi\|_{\D_\alpha} = \sup_{k \in \Z} (1+|k|)^\alpha 
\|\varphi\|_{\rho,2,k,k+1}.
$$ 
In the following we will fix $\alpha>1$ but otherwise arbitrarily small and 
write $\D_\alpha = \D$.
\begin{lemma}
\label{lemma:rough}
For any $\gamma \ge 1/3$, $\rho > 1-\gamma$, $t>s$ and $x,y\in \R^d$, 
there exists a family of random variables $F \in \C$ such that
(\ref{eq:curr-constr-0}) holds with respect to $\W_{[s,t]}^{x,y}$, 
and which satisfy the pathwise bound
\begin{equation}
  \label{eq:F-bound}
 |F_{uv}(\varphi)| \le C_{\gamma,\rho,|t-s|} |u-v|^\gamma (1+N_{[s,t]}(X))^3 
 \|\varphi\|_{\rho,2,s,t}
\end{equation}
for any $\varphi \in \D$, $(u,v) \in \Delta_{[s,t]}$, where
$N_{[s,t]}(X) = (\|X\|_{\gamma,[s,t]} + 
\|\mathbb{X}^2\|_{2\gamma,[s,t]})$ and where
$C_{\gamma,\rho,|t-s|}$ depends only on $\gamma,\rho,|t-s|$.
\end{lemma}

\begin{proof}
Define
\begin{equation*}
F_{uv}(\varphi) = 
\lim_{\delta\tau_\alpha \to 0} 
\sum_\alpha \left(\varphi(\tau_\alpha,X_{\tau_\alpha}) 
+ \nabla \varphi(\tau_\alpha,X_{\tau_\alpha}) 
\mathbb{X}^2_{\tau_\alpha,\tau_{\alpha+1}}\right)
\end{equation*}
where  $\delta\tau_\alpha$ is the mesh of the partition and 
$\mathbb{X}^2_{st}\in C(\R^2;\R^d \times \R^d)$ is the 
twice iterated It\^o integral with respect to $X$. By the results 
on rough path theory in Section \ref{sec:rough-paths} (compare 
Theorem \ref{th:rough-main}), this limit exists whenever 
$N_{[s,t]}(X)<\infty$ and $\|\varphi\|_{2,\rho,s,t}<\infty$. Otherwise 
set $F_{uv}(\varphi) = 0$. Then $F$ is a well defined random variable 
obeying (\ref{eq:F-bound}). Moreover,  $F_{su}(\varphi)+F_{ut}
(\varphi)= F_{st}(\varphi)$ and the locality property for $F$ holds
by definition.

By straightforward estimates we can also prove that 
$\expect_{\W^{x,y}_I} [N_I(X)^3] < \infty$ for any $x,y,I$. Using
this last result, the equivalence between $F$ and the It\^o integral 
for the adapted vector field $\varphi$ can be proved by the same 
approach as the one used in the case of Wiener measure.
\end{proof}

A direct consequence of Lemma~\ref{lemma:rough} is that whenever 
$N_{[s,t]}(X) < \infty$ the map $\varphi \mapsto F(\omega)_{st}
(\varphi)$ can be considered as an element of $\D'$. Moreover if 
we let
$$
\mathcal{N}_{\alpha,p}(X) = \sum_{k \in \Z} (1+|k|)^{-\alpha} 
N_{[k,k+1]}(X)^p,
$$
then whenever $\mathcal{N}_{\alpha,3}(X) < \infty$, the boundary 
currents $C^+_t$ and $C^-_t$ are well defined for any $t$ as elements 
of $\D'$.

\begin{lemma}
\label{lemma:W-lift}
For every bounded $I \subset \R $ there exists a unique Brownian
current $\W^{\sharp,x,y}_I$ on $\Xi_I$. A similar statement holds for
the measures $\W^\sharp$ with first marginal $\W$. Moreover, since
under $\W$ we have  $\mathcal{N}_{\alpha,3}(X) < \infty$ a.s., the
boundary currents are well defined under $\W^\sharp$.
\end{lemma}
\begin{proof}
The existence of the lifted measure for Brownian bridge $\W^{x,y}_I$ 
is essentially contained in Lemma~\ref{lemma:rough}. Its uniqueness 
is a direct consequence of the property~(\ref{eq:curr-constr-0}). 
The proof in the case of $\W$ is similar and we are left to prove that
under $\W^\sharp$ the boundary currents are well defined. The integrals 
in every interval $[i,i+1]$ are well defined, moreover
$$
\left|\int_{i}^{i+1}  \varphi(u,X_u) dX_u\right|\le (1+|i|)^{-\alpha} 
\|\varphi\|_{\mathcal D} [1+N_{[i,i+1]}(X)]^3  
$$
so that the series
$
\int_{t}^{+\infty}\varphi(u,X_u) dX_u =   
\sum_{i>t-1}\int_{i\vee t}^{i+1}  \varphi(u,X_u) dX_u
$
is absolutely convergent if $\mathcal{N}_{\alpha,3}(X) < \infty$.
Under $\W$ we have
$
\expect_{\W} [(1+N_{[i,i+1]}(X))^3|X_0=x] \le C
$ 
uniformly in $i \in \Z$ and $x\in \R^d$ thus
$$
\expect_{\W} [\mathcal{N}_{\alpha,3}(X)|X_0=x]  = 
\sum_{i \in \Z}(1+|i|)^{-\alpha} 
\expect_{\W} [N_{[i,i+1]}(X)^3|X_0=x] 
\le 
C  \sum_{i \in \Z} (1+|i|)^{-\alpha} < \infty
$$
as soon as $\alpha > 1$. This implies that $\mathcal{N}_{\alpha,3}(X)$
is $\W$-a.s. finite.
\end{proof}

\begin{remark}
\rm{
Note that the lifting of a measure from $\X$ to $\Xi$ is in general 
not unique. For instance, we could decide to add some other term to 
the definition of the current,
$$
\tilde C_{st}^X(\varphi) = C_{st}^X(\varphi) + \int_s^t \mathrm{div}\,
\varphi(t, X_t) dt
$$
and obtain a different lifted measure (which is no longer a forward 
current). It would be interesting to explore whether different lifts 
may have a different physical meaning in the models. For instance, 
non-relativistic particles with spin can be (partially) described by 
a current $C^\sigma$ defined as
$$
C^X_{st}(\varphi)[\sigma] = C_{st}^X(\varphi) + \int_s^t \sigma_t 
\cdot \mathrm{curl}\, \varphi(t, X_t) dt
$$
where $(\sigma_t)_t$ is a vector-valued Poisson process describing
the spin of the particle~\cite{Jona,Jona2,HL1}.  
}
\end{remark}

\subsection{It\^o current}
\label{sec:ito-current}
Beside defining Brownian currents, it will be useful for our purposes 
below to define currents for Brownian paths subjected to a potential 
(or penalty function) $V$. The reason is that we need a sufficiently
strong confining mechanism of paths in order to investigate the effect
of a pair interaction on them (given by double stochastic integrals). 
The translation invariant ($V \equiv 0$) regime is presently little
understood. 

While we properly introduce potentials only in the next section, we 
require here that the integral kernel
\begin{equation} 
\label{fkf}
\hat\pi_{t-s}(x,y) = \int e^{- \int_s^t V(X_u)du} \, 
d\W^{x,y}_{[s,t]}(X), \quad \forall s<t \in \R, \;
\forall x,y \in \R^d
\end{equation}
exists. With assumptions on $V$ listed in Section 4.1 below this can be
ensured, and then furthermore the map $(t,x,y) \mapsto \hat \pi_t(x,y)$ 
is jointly continuous and bounded on $(0,\infty) \times \R^d \times 
\R^d$ giving rise to a semigroup $S_t$ via the formula
$$
(S_t f) (x) = \int \hat\pi_{t}(x,y) f(y).
$$
Moreover, for every $t>0$ the semigroup $S_t$ is a bounded operator 
from $L^p$ to $L^q$ for every $ 1 \leq p \leq q \leq \infty$, and by 
the Feynman-Kac formula and the Hille-Yoshida and Stone Theorems it
can be written as $S_t=e^{-tH}$, with $H$ coinciding with the
Schr\"odinger operator $H_0+V$ on $C_0^\infty(\R^d)$. In addition, 
$e^{-tH}f$ is a continuous function for every $f \in L^p, p \in 
[1,\infty]$, $\forall t > 0$. 

Let $\Psi$ be a ground state of the Schr\"odinger operator $H$, i.e., 
a normalized eigenfunction in $L^2(\R^d,dx)$ lying at $\inf\Spec H$. 
Under the conditions given in Section 4.1. this ground state is unique 
and has a strictly positive version. Using this we define the
probability measure $\nu$ on $(\X,\F,\W)$ by
\begin{equation} 
\label{pphi1}
\nu(A) = Z_T^{-1} \int  dx \, \Psi(x) \int dy \, \Psi(y) \int 1_{A}(X) 
e^{-\int_{-T}^T(X_t)dt} \, d\W^{x,y}_T(X)
\end{equation}
with normalizing constant $Z_T$, $A \in \F_{T}$. This measure can be
extended to a measure $\nu$ on the full $\F$ by making use of the 
facts $e^{-tH}\Psi = \Psi$ and $\norm[2]{\Psi} = 1$. The Feynman-Kac 
formula and the Markov property of Brownian motion imply that 
$\{\nu_T\}$ given on $\F_{T}, \,  T>0$, define a consistent family of 
probability measures. In particular, $\nu$ satisfies the DLR equations 
and thus is a Gibbs measure relative to Brownian motion for potential 
$V$; for further details see \cite{BL}. 

Moreover, $\nu$ is the law of a reversible diffusion process with 
stationary distribution $d\omega = \Psi^2 \, dx$ and stochastic 
generator $H_{\omega}$ acting in $L^2(\R^d,d\omega)$ as
$$ 
H_{\omega} f = \frac{1}{\Psi} H (\Psi f) = H_0 f - 
\left(\nabla \log \Psi, {\nabla f}\right)_{\R^d}.
$$
This process is called \emph{It\^o diffusion} (or 
\emph{$P(\phi)_{1}$-process} in quantum field theory). Its transition 
probabilities are given by 
\begin{equation} \label{transition prob}
\nu ( f(X_{t+s}) | X_{s} = x) = \int \pi_{t}(x,y) f(y) \, d\omega(y),
\end{equation}
where
\begin{equation} 
\pi_{t}(x,y) = \frac{\hat\pi_{t}(x,y)}{\Psi(x)\Psi(y)}
\label{pito}
\end{equation}
is the transition density of $\nu$ with respect to its stationary 
distribution. The It\^o process is Markovian, reversible, and has a 
version with continuous paths. Moreover, it is a Brownian
semi-martingale with respect to either the forward or the backward 
filtration, in particular it is the stationary solution of the forward 
stochastic differential equation 
\begin{equation}
  \label{eq:diffusion-eq}
dX_t = \nabla \log \Psi(X_t) dt + dB_t, 
\end{equation}
where $B_t$ is Brownian motion with respect to the forward filtration.

For the It\^o bridge, i.e., the regular conditional probability 
$\nu^{x,y}_T$ of $\nu$ given $X_{-T} = x$, $X_{T} = y$, we will use
the following representation. Take (\ref{pphi1}) describing the density 
of the measure $\nu_T$ with respect to Brownian motion. Then
\begin{equation}
  \label{eq:bridge-2}
\nu_T^{x,y}(A) = \frac{ \Pi_{2T}(x,y)}{Z_T\, \Psi(x) \Psi(y) \pi_{2T}(x,y)}
\expect_{\WW^{x,y}_{[-T,T]}} \left[ 1_A(X) e^{- V_{[-T,T]}(X) }\right],  
\end{equation}
with $\Pi_t$ the Brownian transition kernel (\ref{(2.2)}). This
formula can be checked by noting that 
\begin{equation*}
  \begin{split}
&
\hat \expect_{\nu_T} \left[f(x) g(y) \expect_{\nu_T^{x,y}}[Q] \right] \\ & 
\quad 
= Z_T^{-1} \hat \expect_{\mathcal{W}_{[-T,T]}} \left[f(x)g(y)\Psi(x)\Psi(y) 
\expect_{\WW_{[-T,T]}^{x,y}}[Q(X) e^{- \int_{-T}^TV(X_t)dt}] \right] \\ & 
\quad 
= Z_T^{-1} \expect_{\mathcal{W}_{[-T,T]}} \left[\Psi(X_{-T})\Psi(X_T) 
f(X_T) g(X_{-T}) Q(X) e^{-\int_{-T}^T V(X_t)dt}\right] \\ & 
\quad 
= \expect_{\nu_T} \left[f(X_T) g(X_{-T})Q(X) \right]
\end{split}
\end{equation*}
where $\hat \expect$ denotes the expectation in a new probability
space whose coordinate process is denoted $\hat X$, and with $\hat
X_{-T} = x$, $\hat X_T = y$. 

The assumption that $V$ is Kato-class (see Section 4.1) implies that
the It\^o bridge measure $\nu^{x,y}_{I}$ is absolutely continuous with
respect to the Brownian bridge measure $\W^{x,y}_{I}$. This allows us 
to use the lifting result proved in Lemma~\ref{lemma:W-lift} to show 
that also $\nu^{x,y}_{I}$ allows a lift to the space of currents. 
\begin{lemma}
\label{lemma:nu-lift}
For every bounded $I \subset \R$ there exists a forward current 
$\nu^{\sharp,x,y}_I$ on $\Xi_I$ such that is first marginal is
$\nu^{x,y}_I$. A similar statement holds for the stationary measures 
$\nu^\sharp$ with first marginal $\nu$. Moverover, since under $\nu$ 
we have  $\mathcal{N}_{\alpha,3}(X) < \infty$ a.s., the boundary 
currents are well defined under $\nu^\sharp$.
\end{lemma}
\begin{proof}
We have
$$
\expect_{\nu^{x,y}_I} [N_I(X)^3] \le 
C \left(\expect_{\widehat \W^{x,y}_I} [N_I(X)^6]\right)^{1/2}  
\left(\expect_{\nu^{x,y}_I} [e^{-2\int_I V(X_t)dt}]\right)^{1/2}  
\le C.
$$
Hence the map $F$ defined in Lemma~\ref{lemma:rough} is well defined
and coincides almost surely with the It\^o integral. This allows to
construct the lifted measures $\nu^{\sharp,x,y}_I$ and $\nu^{\sharp}$. 
Moreover, by using stationarity of $\nu$ we have that
$
\expect_{\nu} [N_{[i,i+1]}(X)^3] \le C
$
uniformly in $i\in \Z$, thus the boundary currents are well defined 
under $\nu^\sharp$.
\end{proof}

\section{Gibbs measures on Brownian currents}
\subsection{Conditions on the potentials}
\label{sec:potentials}
We use the same terminology of the usual DLR theory and introduce 
``potentials" and ``energy functionals" below. 

An \emph{external potential} is a Lebesgue measurable function $V: 
\R^d \to \R$ that we will choose from the Kato class, i.e., an element 
of the space $\mathcal{K}(\R^d)$ defined by the condition
\begin{equation}
\lim_{r \to 0} \sup_{x \in \R^d} \int_{B_r(x)} 
|g(y-x)V(y)| \, dy = 0,
\end{equation}
with $B_r(x)$ the ball of radius $r$ centered at $x$, and 
\begin{equation}
 g(x) = \left\{ %
\begin{array}{ll} 
        |x| & \mbox{if} \;\; d = 1 \\         
        -\ln|x| & \mbox{if} \;\; d = 2 \label{k2}\\
        |x|^{2-d} & \mbox{if}   \;\;  d \geq 3.
    \end{array} \right. 
\end{equation}
This space is large enough to contain many choices of interest,
while allowing the Feynman-Kac formula for the Schr\"odinger
semigroup $e^{-tH}$, $t \geq 0$, to hold. This is generated by  
the Schr\"odinger operator $H = H_0 + V$ defined on 
$L^2(\R^d,dx)$ as a form sum ($V$ regarded as a multiplication 
operator). For Kato-class potentials $H$ is essentially 
self-adjoint on the form core $C^\infty_0(\R^d)$. In addition, 
we will require of $V$ to be such that 
\begin{enumerate}
\item[(1)]
$H$ has a unique strictly positive eigenfunction (ground state) 
$\Psi$ at $E = \inf\Spec H$, with the property that $\Psi \in
L^1 \cup L^\infty$;
\item[(2)] 
$e^{-tH}$ is intrinsically ultracontractive.
\end{enumerate}
Recall the meaning of the latter property. Write $d\omega = \Psi^2 
dx$ on $\R^d$ as before, and define the isometry (\emph{ground 
state transform}) $j: L^2(\R^d,d\omega) \rightarrow L^2(\R^d,dx)$, 
$f \mapsto \Psi f$. Then $D(H_\omega) = j^{-1}D(H)$ and $H_\omega 
f=(j^{-1}Hj)f = (1/\Psi) H(\Psi f) = -(1/2)\Delta f-({\nabla\ln\Psi},
\nabla f)_{\R^d}$, for every $f \in D(H_\omega)$. The associated 
semigroup $e^{-tH_\omega}$ exists for all $f\in L^2(\R^d,d\omega)$ 
and $t \geq 0$. $e^{-tH}$ is \emph{intrinsically ultracontractive} 
when $e^{-tH_\omega}$ is \emph{ultracontractive}, i.e., it maps 
$L^2(\R^d, d\omega)$ into $L^\infty(\R^d,d\omega)$ continuously. 
Equivalently, this means that $\|{e^{-tH_\omega}}\|_{2,\infty} < 
\infty$, $\forall t> 0$, and it is a monotonically decreasing 
function in $t$. Moreover, the integral kernel (\ref{pito}) of 
$e^{-tH_\omega}$ satisfies $0 \leq \pi_t(x,y) \leq 
\|e^{-(t/2)H_\omega}\|_{2,\infty}^2$ almost surely.

These conditions are in particular satisfied for $V$ bounded 
from below, continuous, and sufficiently confining, i.e., for 
which there exist constants $C_1, C_2 > 0$, $C_3, C_4\in \R$, 
and $a,b$ with $2 < a < b < 2a - 2$ such that the positive 
part of the potential, $V^+ = \sup \{0,V\}$ satisfies
\begin{equation} \label{iuccond}
C_1|x|^a + C_3 \leq V^{+}(x) \leq C_2|x|^b + C_4.
\end{equation}

\bigskip

A \emph{pair interaction potential} is a Lebesgue measurable 
function $W: \R^d \times \R \to \R$, even in both of its 
variables, which we require to 
\begin{enumerate}
\item[(1)]
have positive Fourier transform;
\item[(2)]
satisfy the regularity condition that there exists 
$M_{I,\beta} \in \R$ such that
\begin{equation}
\sup_{x\in \R^d ,t\in I} \|W(x,t)\|_{\D_{\beta}} \leq
M_{I,\beta},
\label{a22}
\end{equation}
for $\beta > \max\{\alpha, 3\}$ and every bounded $I \subset 
\R$.
\end{enumerate}
The requirement $\beta>\alpha$ is needed for having a well 
defined interaction energy (actually in our applications for 
any $\beta>1$ there is a suitable $\alpha$ so that this holds), 
while $\beta>3$ is a decay condition sufficient for ensuring 
the convergence of the cluster expansion in Chapter 5 below. 
 
An example satisfying these conditions is $W^\rho$ seen in 
Section 1.

\bigskip
Finally we write down the energies appearing in the definition of 
the densities of Gibbs measures we are going to study. With given 
pair potential $W$, for all $a,b \in \mathcal{D}'$ consider the 
function $\psi_{k,\varpi}(x,t) = e^{ik \cdot x +i\varpi t}\in 
\mathcal{D}$ and define the (possibly unbounded) quadratic form 
$$
\langle a,b \rangle_W := \int \wW(k,\varpi) 
a(\psi_{k,\varpi})b(\overline{\psi_{k,\varpi}}) dk d\varpi. 
$$
By using the quadratic form, 
for all bounded $I \subset \R$ and every $\mathbb{X} = (X, C^X), 
\, \mathbb{Y} = (Y, C^Y) \in \Xi$ define the \emph{internal energy 
functional}
\begin{equation}
H_I(\mathbb{X}) = V_I(X)  + \frac{\lambda }{2}\langle 
C^X_I, C^X_I \rangle_W,
\label{hin}
\end{equation}
and 
\emph{interaction functional}
\begin{equation}
H_I(\mathbb{X}|\mathbb{Y}) =  V_I(X) + 
\frac{\lambda}{2}\langle C^X_I, C^X_I \rangle_W + 
\lambda  \langle C^X_I, C^Y_{I^c} \rangle_W,
\label{hint}
\end{equation}
with parameter $\lambda \in \R$, where we wrote $V_I(X) = 
\int_I V(X_t) dt$.

\subsection{Gibbs specifications}

On $\Xi$ with its associated $\sigma$-Borel field $\A$ we take now 
$\W^\sharp$ as reference measure and define a Gibbs specification.
\begin{definition}[Gibbs specification]
Take the regular version $\W_I^\sharp(d\mathbb{X}|\mathbb{Y})$ of 
the measure $\W^\sharp$ conditional on $\mathbb{Y}$ in $\A_{I^c}$, 
and $H_I(\mathbb{X})$ given by (\ref{hin}). We call the family of
probability kernels $\{\mu^\sharp_I\}_I$ on $\Xi$ indexed by the 
bounded intervals $I \subset \R$, 
\begin{equation}
\label{eq:curr-finite-vol}
\mu^\sharp_I(d\mathbb{X}) = \frac{e^{-H_I(\mathbb{X})}}
{Z_I} \W^\sharp_I(d\mathbb{X})  
\end{equation}
a \emph{Gibbs specification on Brownian currents with free 
boundary condition}. Take $H_I(\mathbb{X}|\mathbb{Y})$ given 
by (\ref{hint}). We call the family $\{\rho^\sharp_I\}_I$ on
$\Xi$,
\begin{equation}
\label{eq:spec}
\rho^\sharp_I(d\mathbb{X}|\mathbb{Y}) = 
\frac{e^{-H_I(\mathbb{X}|\mathbb{Y})}}{Z_I(\mathbb{Y})} 
\W^\sharp_I(d\mathbb{X} | \mathbb{Y})  
\end{equation}
a \emph{Gibbs specification on Brownian currents with boundary
condition $\mathbb{Y}$}. 
\end{definition}

\begin{definition}[Gibbs measure]
A probability measure $\mu$ on $(\Xi,\A,\W^\sharp)$ is a 
\emph{Gibbs measure} for the potentials $V$ and $W$ if it is 
consistent with the specification $\{\rho^\sharp_I\}_I$, i.e.,
there exists a version of its conditional probabilities with
respect to the family $\{\A_{I^c}\}_I$ which agrees with 
$\{\rho^\sharp_I\}_I$ for all bounded $I \subset \R$. 
\end{definition}
In the following chapter our main concern will be to prove the
existence of such Gibbs measures. 

On Gibbs specifications here is a first result. 
\begin{lemma}
The family $\{ \rho_I^\sharp \}_I$ is consistent, i.e., for 
every pair of bounded intervals $I \subset J \subset \R$ 
we have $\int\int F(\mathbb{X})\rho_I^\sharp(d\mathbb{X}|\mathbb{Y})
\rho_J^\sharp(d\mathbb{Y}|\ZZ) = \int  F(\mathbb{X})
\rho_J^\sharp(d\mathbb{X}|\ZZ)$, for any bounded measurable $F:
\Xi \to \R$.
\end{lemma}
\begin{proof}
The family  $\{ \W_I^\sharp \}_I$ is consistent by its definition. 
Hence
\begin{equation*}
  \begin{split}
\int \int & F(\mathbb{X})
\rho_I^\sharp(d\mathbb{X}|\mathbb{Y})
\rho_J^\sharp(d\mathbb{Y}|\ZZ) 
\\ & =
     \int \int F(\mathbb{X}) \frac{e^{-V_I(X)-(\lambda/2)
     \langle C^X_I, C^X_I \rangle_W - 
     \lambda\langle C^X_I,C^Y_{I^c}\rangle_W}}{Z_I(\mathbb{Y})}
\\ & \quad \quad \times
 \frac{e^{-V_J(Y)-(\lambda/2)\langle C^Y_J, C^Y_J \rangle_W-
 \lambda\langle C^Y_J,C^Z_{J^c} \rangle_W}}{Z_J(\ZZ)}
\W_I^\sharp(d\mathbb{X}|\mathbb{Y})\W_J^\sharp(d\mathbb{Y}|\ZZ) 
\\ &=  
     \int \int F(\mathbb{X}) \frac{e^{-V_I(X)-(\lambda/2)
     \langle C^X_I, C^X_I \rangle_W-
     \lambda\langle C^X_I,C^Y_{I^c}\rangle_W}}{Z_I(\mathbb{Y})}
\\ & \quad \quad \times
{e^{-V_I(Y)-(\lambda/2)\langle C^Y_I, C^Y_I \rangle_W
-\lambda \langle C^Y_I,C^Z_{J^c} \rangle_W -
\lambda\langle C^Y_I,C^Y_{K} \rangle_W }}{}
\\ & \quad \quad \times
 \frac{e^{-V_K(Y)-(\lambda/2)\langle C^Y_K, C^Y_K \rangle_W-
 \lambda\langle C^Y_K,C^Z_{J^c} \rangle_W}}{Z_J(\ZZ)}
\W_I^\sharp(d\mathbb{X}|\mathbb{Y})\W_J^\sharp(d\mathbb{Y}|\ZZ) 
  \end{split}
\end{equation*}
where we split off $C^Y_J = C^Y_I + C^Y_K$ with $K = J\backslash 
I$. This gives for the right hand side
\begin{equation*}
  \begin{split}
\quad  &
     \int \int F(\mathbb{X}) \frac{e^{-V_I(X)-(\lambda/2)
     \langle C^X_I,C^X_I \rangle_W-
     \lambda\langle C^X_I,C^Y_{I^c}\rangle_W}}{Z_I(\mathbb{Y}_2)}
\\ & \quad \quad \times
\left( \int {e^{-V_I(Y_1)-(\lambda/2)\langle C^{Y_1}_I,C^{Y_1}_I
\rangle_W-\lambda\langle C^{Y_1}_I,C^Z_{J^c} \rangle_W -
\lambda\langle C^{Y_1}_I,C^{Y_2}_{K} \rangle_W }}{} 
 \W_I^\sharp(d\mathbb{Y}_1|\mathbb{Y}_2)\right)
\\ & \quad \quad \times
 \frac
{e^{-V_K(Y_2)-(\lambda/2)\langle C^{Y_2}_K, C^{Y_2}_K \rangle_W-
\lambda\langle C^{Y_2}_K,C^Z_{J^c} \rangle_W}}{Z_J(\ZZ)}
\W_I^\sharp(d\mathbb{X}|\mathbb{Y}_2) \W_K^\sharp(d\mathbb{Y}_2|\ZZ), 
  \end{split}
\end{equation*}
where we used the fact that 
$\int \W_I^\sharp(d\mathbb{Y}_1|\mathbb{Y}_2)\W_K^\sharp
(d\mathbb{Y}_2|\ZZ) = \W_J^\sharp(d\mathbb{Y}_1|\ZZ)$. Note that 
the expression between the brackets equals $Z_I(\mathbb{Y}_2)$, 
thus we further obtain
\begin{equation*}
  \begin{split}
\quad & 
     \int \int F(\mathbb{X}) e^{-V_I(X)-(\lambda/2)
     \langle C^X_I, C^X_I \rangle_W-
     \lambda\langle C^X_I,C^Y_{I^c} \rangle_W}
\\ & \quad \quad \times
 \frac
{e^{-V_K(Y_2)-(\lambda/2)\langle C^{Y_2}_K, C^{Y_2}_K \rangle_W -
\lambda\langle C^{Y_2}_K,C^Z_{J^c} \rangle_W}}{Z_J(\ZZ)}
\W_I^\sharp(d\mathbb{X}|\mathbb{Y}_2) \W_K^\sharp(d\mathbb{Y}_2|\ZZ) 
\\ \quad &=  
     \int  F(\mathbb{Y})
\frac{e^{-V_J(Y)-(\lambda/2)\langle C^{Y}_J, C^{Y}_J \rangle_W -
\lambda\langle C^{Y}_J,C^Z_{J^c} \rangle_W}}{Z_J(\ZZ)}
\W_J^\sharp(d\mathbb{Y}|\ZZ) = \int F(\mathbb{Y}) 
\rho_J(d\mathbb{Y}|\ZZ).
  \end{split}
\end{equation*}
\end{proof}

The forward current $\W^\sharp$ has the key property that
$
C^X_{st}(\varphi) = \int_s^t \varphi(u,X_u) dX_u
$,  
$\W^\sharp$-a.s. for all $(s,t)\in \Delta$ and all adapted 
$\varphi\in \D$. This will enable us to show that the finite 
volume measures $\mu_I$ coincide with the marginals of the 
measures $\mu^\sharp_I$ on the first component of the product 
$\Xi$. The specification $\{\rho^\sharp_I\}$ can then be 
considered as a suitable rigorous replacement for the DLR 
description of the infinite-volume limit. A Gibbs measure on 
$\X$ will then be a measure for which there exists a unique 
lift to the space $\Xi$ of currents satisfying the relation~
(\ref{eq:curr-constr-0}) ensuring the identification of the 
current with the stochastic integral and which satisfy the 
DLR conditions with respect to the specification 
$\{\rho^\sharp_I\}$.

To show that the specification is well defined we rewrite the 
various terms using the fact that, under the measure 
$\W^\sharp_I(d\mathbb{X}|\mathbb{Y})$ we have pathwise 
equality between the current $C^X$ and the stochastic integral 
with respect to $X$ for adapted integrands belonging to $\D$. 
Then
\begin{equation*}
  \begin{split}
\langle C_I^X, C_I^X \rangle_W  & = \int \wW(k,\varpi)
C_I^X(\psi_{k,\varpi}) C_I^X(\overline{\psi_{k,\varpi}})
dk d\varpi
   \\ &   = \int \wW(k,\varpi) |C_I^X(\psi_{k,\varpi})|^2
             dk d\varpi 
   \\ &   = \int \wW(k,\varpi)
\left|\int_I \psi_{k,\varpi}(t,X_t) dX_t \right|^2 
dk d\varpi = 2 W_I(X)
  \end{split}
\end{equation*}
and
\begin{equation*}
  \begin{split}
\langle C_I^X, C_{I^c}^Y \rangle_W & = 
\int \wW(k,\varpi) C_I^X(\psi_{k,\varpi}) C_{I^c}^Y
(\overline{\psi_{k,\varpi}}) dk d\varpi 
\\
 & =
C_I^X\left( \int \wW(k,\varpi)\psi_{k,\varpi} 
C_{I^c}^Y(\overline{\psi_{k,\varpi}}) dk d\varpi \right)
\\
 & =
C_I^X\left( w^{C_{I^c}^Y}\right) = \int_{-T}^{T}
C_{I^c}^Y(W(\cdot - X_s,\cdot - s)) dX_s,
  \end{split}
\end{equation*}
with
$$
w^C(x,t) = \int \wW(k,\varpi)\psi_{k,\varpi}(t,x)
C(\overline{\psi_{k,\varpi}}) dk d\varpi  = 
C(W(x-\cdot,t-\cdot)).
$$
By using these equalities it is seen that the 
specification~(\ref{eq:spec}) takes the form
\begin{equation}
\label{eq:spec2}
\rho^\sharp_I(d\mathbb{X}|\mathbb{Y}) =
\frac{e^{-V_I(X)-\lambda W_I(X)-
\lambda \int_I w^{C_{I^c}^Y}(u,X_u)dX_u}}{Z_I(\mathbb{Y})} 
\W^\sharp_I(d\mathbb{X} | \mathbb{Y})  
\end{equation}
and it is well defined as soon as the exponential weight is integrable
and the integral is different from zero. The conditions on $V$ and $W$ 
make sure this is true. Indeed, for Kato-class potentials exponential 
integrability is a consequence of Khasminskii's Lemma \cite{Sim82}. 
Moreover, since the Fourier transform of $W$ is positive by assumption 
and $\lambda >0$, the internal energy term is negative and thus 
exponentially integrable without any further restriction. For the 
interaction with the boundary current we have
$$
|w^{C_{I^c}^Y}(x,t)| = |C_{I^c}^Y(W(x,t))| \le M_I \|C_{I^c}^Y\|_{\D'} 
$$
with $M_I = \sup_{x\in \R^d,t\in I} \|W(x,t)\|_{\D}$, which by
condition (2) on $W$ is finite. Hence the stochastic integral in the 
exponent has a bounded and adapted integrand and thus by standard 
techniques it follows that it is exponentially integrable for any
value of $\lambda$.

\medskip
By making use of the It\^o current defined in
Section~\ref{sec:ito-current}, the specification $\{\rho^\sharp_I\}_I$ 
can be finally written as
\begin{equation}
\label{eq:specc2}
\rho^\sharp_I(d\mathbb{X}|\mathbb{Y}) =
\frac{e^{-\lambda W_I(X)-
\lambda \int_{-T}^T w^{C_{I^c}^Y}(u,X_u)dX_u}}{\mathcal Z_I(\mathbb{Y})} 
\nu^\sharp_I(d\mathbb{X} | \mathbb{Y}) . 
\end{equation}
Note that this is a forward current on $\Xi$ but by the above results it 
can be obtained as the unique lift of its marginal on $\X$ satisfying the
identification~(\ref{eq:curr-constr-0}) between currents and stochastic
integrals. 

\begin{remark}
\rm{
The specification~(\ref{eq:spec2}) seems to depend only on the path
and the currents appearing in the definition of the vector-field 
$w^C$ that describes the interaction with the boundary paths. The point 
of introducing measures and specifications on currents resides in the 
fact that we are not able to describe (\ref{eq:spec2}) in terms of 
paths alone. The framework of stochastic currents is not the only
possibility to solve this difficulty. A different way to proceed is 
considering directly rough paths and defining the measures and 
specifications on the space of (step-2) rough paths, i.e., formally of
couples $(X,\mathbb{X}^2)$, where $\mathbb{X}^2$ is the twice iterated
integral associated with the paths $X$. This would solve the problem
of stochastic integrals, which can then be defined as rough integrals, 
and with suitable growth conditions on the rough paths we would be 
also allowed to define the interaction terms with boundary paths (over
unbounded time intervals) and specifications similarly to that on the 
currents. Our approach is motivated essentially by the consideration 
that currents are more basic objects than rough paths. We prefer to
see rough path theory as a tool for obtaining stochastic currents in
useful topologies. Indeed, in principle the construction of good 
versions of stochastic integrals can be carried out without recourse 
to rough paths~\cite{FGGT,FGR}. 
}
\end{remark}

\section{Existence of Gibbs measures for Brownian currents}
\label{sect:cluster}

\subsection{Cluster representation}
\label{sec:cluster-rep}

In the following we will construct a Gibbs measure that is consistent
with the specification $\{\rho^\sharp_I\}_I$. This will be achieved 
by breaking up paths according to a sequence of bounded subintervals 
of the real line, and constructing Gibbs measures for bounded intervals. 
Taking limits over these Gibbs measures will result in a Gibbs measure 
on $\X$ whose lifted measure to $\Xi$ is consistent with the given
specification. As mentioned before, a reasonably confining $V$ is 
needed to make sure that the paths are not allowed to escape to 
infinity with large probability. 

The following notion of convergence will be used below to discuss
Gibbs measures. Let generally $E$ 
be a metric space, and $C(\R,E)$ the space of continuous paths
$\{X_t; t \in \R\}$ with values in $E$. For any bounded interval $I  
\subset \R$ let $\E_I \subset \E$ be a sub-$\sigma$-field of the Borel 
$\sigma$-field $\E$ of $E$ generated by the evaluations $\{X_t: t \in I\}$. 
A sequence of probability measures $(m_n)_{\{n \in \mathbb{N}\}}$ on 
$C(\R,E)$ is said to converge locally weakly to the probability measure 
$m$ if for any such $I$ the restrictions $m_n|_{\E_I}$ converge weakly 
to the measure $m|_{\E_I}$. 

The main result of this paper is the following

\begin{theorem}
Suppose $V$ and $W$ satisfy the assumptions stated in Section 
\ref{sec:potentials}. Take any unbounded increasing sequence 
$(T_n)_{n\geq 0}$ of positive real numbers, and suppose $0 < |\lambda| 
\leq \lambda^*$ with $\lambda^*$ small enough. Then  the local weak 
limit $\lim_{n \rightarrow \infty} \mu_{T_n} = \mu$ exists on $\X$ 
and does not depend on the choice of sequence $T_n$. Its unique lift 
$\mu^\sharp$ on $\Xi$ is a Gibbs probability measure consistent with
the specification  $\{\rho^\sharp_I\}_I$.
\label{exist}
\end{theorem}
\begin{proof}
We develop a cluster expansion, i.e., choose the coupling parameter 
$\lambda$ sufficiently small for being able to control the measure 
for the interaction switched on ($\lambda \neq 0$) in terms of a 
convergent perturbation series around the free case ($\lambda = 0$).  
The theorem follows then through Propositions \ref{p0}, \ref{p2},
\ref{pclustest} and \ref{prop:infinite-measure} below. 
\end{proof}

\medskip
Take a division of $[-T,T]$ into disjoint intervals $\tau_k = (t_k, 
t_{k+1})$,  $k = 0,..., N-1$, with $t_0 = -T$ and $t_N = T$, each of 
length $b$, i.e. fix $b = 2T/N$; for convenience we choose $N$ to be 
an even number so that the origin is endpoint to some intervals. We 
break up a path $X$ into pieces $X_{\tau_k}$ by restricting it to 
$\tau_k$. The total energy contribution of the pair interaction then
becomes
\begin{equation}
W_T(X) = \frac{1}{2} \sum_{i,j = 0}^{N-1} 
\langle C^X_{\tau_i},C^X_{\tau_j}\rangle_W  = 
\sum_{0 \leq i < j \leq N-1} W_{\tau_i,\tau_j}
\label{sum}
\end{equation}
where with the notation $\J_{ij} = 
\langle C^X_{\tau_i},C^X_{\tau_j}\rangle_W$ we have
\begin{equation}
W_{\tau_i,\tau_j} 
= \left\{ 
\begin{array}{ll}
\vspace{0.2cm} 
\J_{ij} + \J_{ji} & \mbox{if $|i-j| \geq 2$} \nonumber \\
\vspace{0.2cm} 
\frac{1}{2} (\J_{ii} + \J_{jj}) + \J_{ij} + \J_{ji}  
& \mbox{if $|i-j| = 1$, and 
$i \neq 0$, $j \neq N - 1$} \nonumber \\
\vspace{0.2cm} 
\J_{ij} + \J_{ji} + \frac{1}{2} \J_{00} 
& \mbox{if $i = 0$ and $j = 1$} \nonumber \\ 
\vspace{0.2cm} 
\J_{ij} + \J_{ji} + \frac{1}{2} \J_{N-1 \; N-1} 
& \mbox{if $i = N-1$ and $j = N-2$.} \nonumber \\
\end{array} \right.
\end{equation}
To keep the notation simple we do not make explicit the $X$ 
dependence in these objects.

By using (\ref{sum}) we obtain
\begin{equation}
e^{-\lambda W_T} = \prod_{0 \leq i < j \leq N-1} 
(e^{-\lambda W_{\tau_i,\tau_j}} + 1 - 1) =  1 + \sum_{{\cal R} \neq 
\emptyset} \prod_{(\tau_i,\tau_j) \in {\cal R}} 
(e^{-\lambda W_{\tau_i,\tau_j}} - 1).
\label{sum1}
\end{equation}
Here the summation is performed over all nonempty sets of different 
pairs of intervals, i.e. ${\cal R} = \{(\tau_i,\tau_j): 
(\tau_i, \tau_j) \neq (\tau_{i'},\tau_{j'}) \; \mbox{whenever} \; (i,j) 
\neq (i',j') \}$. 

A break-up of the paths involves a corresponding factorization of the 
reference measure into It\^o bridges for each subinterval. Put $X_{t_k} = 
x_k$ for the positions at the time-points of the division, $\forall 
k = 0,...,N$, with $-T = t_0 < t_1 < ... < t_N = T$. We write for a 
shorthand
\begin{equation}
\nu_T^{\mathbf{x}}(\cdot) \equiv \nu_T(\, \cdot \,| X_{t_0} = x_0, \ldots , 
X_{t_N} = x_N) = \prod_{k=0}^{N-1} d\nu_{\tau_k}^{x_k,x_{k+1}}(\cdot).
\label{mar}
\end{equation}
Let $p_{t_0,...,t_N} (x_0,...,x_N)$ be the density with respect to 
$\prod_{k=0}^N d\omega(x_k)$, $d\omega = \Psi^2 dx$, of the joint 
distribution of positions of the path $X$ recorded at the time-points 
of the division. 
By Markovianness it follows that
\begin{eqnarray}
p_{t_0,...,t_N} (x_0,...,x_N) 
& = &
\prod_{k=0}^{N-1} \pi_b(x_{k+1},x_k) = \prod_{k=0}^{N-1} 
(\pi_b(x_{k+1},x_k) - 1 + 1)  
\nonumber \\ 
& = &
1 + \sum_{\cal S \neq \emptyset} \prod_{k: \tau_k \in {\cal S}} 
(\pi_b(x_{k+1},x_k) - 1),
\nonumber
\end{eqnarray}
where $\pi_t$ is the transition kernel for the It\^o diffusion given by
(\ref{pito}). The summation runs over all nonempty sets ${\cal S} = 
\{\tau_k = (t_k,t_{k+1})\}$ of different pairs of consecutive time-points. 

In order to have a systematic control over these sums we introduce:

\begin{itemize}
\item[(1)] 
 \emph{Contours.} \hspace{0.2cm} 
Two distinct pairs of intervals $(\tau_i,\tau_j)$ and 
$(\tau_{i'},\tau_{j'})$ will be called directly connected and denoted 
$(\tau_i,\tau_j) \sim (\tau_{i'},\tau_{j'})$ if one interval of the pair 
$(\tau_i,\tau_j)$ coincides with one interval of the pair $(\tau_{i'},
\tau_{j'})$. A set of connected pairs of intervals is a collection 
$\{(\tau_{i_1},\tau_{j_1}),..., (\tau_{i_n},\tau_{j_n})\}$ in which 
each pair of intervals is connected to another through a sequence of 
directly connected pairs, i.e., for any $(\tau_i,\tau_j) \neq 
(\tau_{i'},\tau_{j'})$ there exists $\{(\tau_{k_1},\tau_{l_1}),..., 
(\tau_{k_m},\tau_{l_m})\}$ such that $(\tau_i,\tau_j) \sim 
(\tau_{k_1},\tau_{l_1}) \sim ... \sim (\tau_{k_m},\tau_{l_m}) 
\sim (\tau_{i'},\tau_{j'})$. A maximal set of connected pairs of 
intervals is called a contour, denoted by $\gamma$. We denote by 
$\bar\gamma$ the set of all intervals that are elements of the pairs 
of intervals belonging to contour $\gamma$, and by $\ga^*$ the set of 
time-points of intervals appearing in $\bar\ga$. Two contours $\ga_1, 
\ga_2$ are disjoint if they have no intervals in common, i.e. 
$\bar \ga_1 \cap \bar \ga_2 = \emptyset$. Clearly, $\cal R$ can 
be decomposed into sets of pairwise disjoint contours: 
${\cal R} = \cup_{r \geq 1} {\cal R}_r$, where ${\cal R}_r = 
\{\ga_1,...,\ga_r\}$ with $\bar\ga_i \cap \bar \ga_j = \emptyset$, 
$i \neq j$; $i,j = 1,...,r$. 

\item[(2)]  \emph{Chains.} \hspace{0.2cm} 
A collection of consecutive intervals $\{\tau_{j},\tau_{j+1}...,
\tau_{j+k} \}$, $j \geq 0$, $j+k \leq N-1$ is called a chain. 
As in the case of contours, $\bar \rh$ and $\rh^*$ mean the set of 
intervals belonging to the chain $\rh$ and the set of time-points in 
$\rh$, respectively. Two chains $\rh_1, \rh_2$ are called disjoint if 
they have no common time-points, i.e. $\rh_1^* \cap \rh^*_2 = \emptyset$. 
Denote by $\partial^-\rh$ resp. $\partial^+\rh$ the leftmost resp. 
rightmost time-points belonging to $\rh$. 

\item[(3)]
 \emph{Clusters.} \hspace{0.2cm} 
Take a (non-ordered) set of disjoint contours and disjoint chains, 
$\Ga = \{\gamma_1,...,\gamma_r;\varrho_1,...,\varrho_s\}$, with some 
$r \geq 1$ and $s \geq 0$. Note that such contours and chains may have 
common time-points. The notation $\Ga^* = (\cup_i \ga^*_i) \cup 
(\cup_j \rh_j^*)$ means the set of all time-points appearing as
beginnings or ends of intervals belonging to some contour or chain 
in $\Gamma$. Also, we put $\bar \Ga = (\cup_i \bar\ga_i) \cup (\cup_j 
\bar\rh_j)$ for the set of intervals appearing in $\Ga$ through 
entering some contours or chains. $\Gamma$ is called a cluster if 
$\Ga^*$ is a connected collection of sets (in the usual sense), and 
for every $\varrho \in \Gamma$ we have that $\partial^-\rh, 
\partial^+\rh \in \cup_{j=1}^r\ga^*_j$. This means that in a cluster 
chains have no loose ends. We denote by $\K_N$ the set of all clusters 
for a given $N$. 

\end{itemize}

With these notations the sum in (\ref{sum1}) is then further expanded as
\begin{equation}
\sum_{{\cal R} \neq \emptyset} \prod_{(\tau_i,\tau_j) \in \cal R} 
(e^{-\lambda W_{\tau_i,\tau_j}} - 1) = \sum_{r \geq 1} 
\sum_{\{\gamma_1,...,\gamma_r\}} \prod_{k = 1}^r \prod_{(\tau_i,\tau_j) 
\in \gamma_k} (e^{-\lambda W_{\tau_i,\tau_j}} - 1)
\label{sum2}
\end{equation}
where now summation goes over collections $\{\gamma_1,...,\gamma_r\}$ 
of contours such that $\bar\ga_k \cap \bar\ga_{k'} = \emptyset$ 
unless $k = k'$. In a similar way (\ref{mar}) appears in the form
\begin{equation}
\sum_{\cal S \neq \emptyset} \prod_{k: \tau_k \in {\cal S}} 
\left( \pi_b(x_{k+1}, x_k) - 1 \right) = 
\sum_{s \geq 1} \sum_{\{\varrho_1,...,\varrho_s\}} \prod_{j=1}^s 
\prod_{k: \tau_k \in \varrho_j} \left(\pi_b(x_{k+1},x_k) - 1 \right).
\label{sum3}
\end{equation}
Here $\{\varrho_1,...,\varrho_s\}$ is a collection of disjoint chains,
and this formula justifies how we defined them. 

For every cluster $\Ga = \{\ga_1,...,\ga_r;\rh_1,...,\rh_s\} \in \K_N$ 
define the function
\begin{equation}
\kappa_\Ga = \prod_{l=1}^r \prod_{(\tau_i,\tau_j) \in \gamma_l} 
(e^{-\lambda W_{\tau_i,\tau_j}} - 1) \prod_{m=1}^s 
\prod_{k: \tau_k \in \rh_m}\left(\pi_b(x_{k+1},x_k) - 1\right). 
\label{kappa}
\end{equation}
Also, introduce the auxiliary probability measure 
\begin{equation}
d\chi_N(X) = \prod_{k=0}^{N-1} d\nu_{\tau_k}^{x_k,x_{k+1}}(X_{\tau_k}) 
\prod_{k=0}^N d\omega(x_k),
\label{auxi}
\end{equation}
and look at 
\begin{equation}
K_{\Ga} = \mathbb{E}_{\chi} [\kappa_\Ga],
\label{aux}
\end{equation}
where $\chi$ is the unique extension over the real line of the family 
of consistent probabilities $\{\chi_N\}_{N\ge 1}$. Note that 
$\int(\pi_b(x_{k+1},x_k)-1) d\omega(x_{k+1}) = \int(\pi_b(x_{k+1},x_k) 
-1) d\omega(x_k) = 0$. This is the reason why from a cluster we rule 
out chains having loose ends; for any such chain $\mathbb{E}_{\chi_N} 
[\kappa_\Ga] = 0$.

Define 
\begin{equation*}
\phi^T(\Ga_1,...,\Ga_n) = \begin{cases}
1 & \text{if $n=1$}\\
\sum_{G \in \G^n} \prod_{\{i,j\}\in G} 
(- 1_{\Ga^*_i \cap \Ga^*_j \neq \emptyset}) & \text{if $n>1$},
\end{cases}
\end{equation*}
with $\G^n$ the set of connected graphs on the vertex set
$\{1,\dots,n\}$. Note that $\phi^T(\Ga_1,...,\Ga_n)=0$ if the graph 
on the vertex set  $\{1,\dots,n\}$ with edges $\{i,j\}$ drawn 
whenever $\Gamma^*_i \cup \Gamma^*_j \neq \emptyset$, is connected.

By putting (\ref{sum2}), (\ref{mar}), (\ref{sum3}), (\ref{kappa}) and 
(\ref{aux}) together we obtain the cluster representation of the
partition function $Z_T$.
\begin{proposition}
\label{p0}
For every $T = Nb/2 > 0$ we have 
\begin{equation}
\label{clexp}
Z_T  = 1 + \sum_{n \geq 1} \sum_{\{\Ga_1,...,\Ga_n\} \in \K_N
\atop \Ga^*_i \cap \Ga^*_j = \emptyset, i \neq j} \prod_{l=1}^n 
K_{\Ga_l}.
\end{equation}
If the activities $K_{\Ga}$ satisfy the bound
\begin{equation}
\sum_{\Ga \in \K_N \atop \Ga^* \ni 0, |\bar\Ga| = n} |K_{\Ga}| 
\;\leq\;  c \; \eta^n
\label{conv}
\end{equation}
for $\eta > 0$ small enough, then the series above and at the right 
hand side of
\begin{equation}
\log Z_T = \sum_{n \geq 1} \sum_{\{\Ga_1,...,\Ga_n\} \in \K_N 
\atop 0 \in \Ga_1^*} \phi^T(\Ga_1,...,\Ga_n) \prod_{l=1}^n K_{\Ga_l}
\end{equation}
are absolutely convergent, uniformly in $N$, and the latter one 
gives the logarithm of the partition function for the interval 
$[-T,T]$.
\end{proposition}

\noindent
The expression of the logarithm  and the absolute convergence of the 
sums are a general result of cluster expansion techniques, for details 
of proof see~\cite{MM}.

\subsection{Convergence of cluster expansion}

\begin{proposition}
\label{p2}
Suppose that there exist a function $D: \mathbb{N}\times\mathbb{N} 
\rightarrow (0,\infty)$ and numbers $\e, C>0$ with
$
\sup_{i\in\mathbb{N}} \sum_{j\in\mathbb{N}} D(i,j) \le C
$
such that for every $N > 0$ and every cluster 
$\Ga = \{\ga_1,...,\ga_r;\rh_1,...,\rh_s\} \in \K_N$, the bound 
\begin{equation}
|K_{\Gamma}| \; \leq \; \left(\prod_{l=1}^r \prod_{(\tau_i,\tau_j) 
\in \gamma_l} \e D(i,j)\right) \e^{\sum_{m=1}^s |\bar\rh_m|}
\label{clustest-abs}
\end{equation}
holds. Then there is a constant $c > 0$ and a function 
$0 < \eta(\eps) < 1$ with $\eta \to 0$ as $\eps \to 0$ such that
\begin{equation}
\sum_{\Ga \in \K_N \atop \Ga^* \ni 0, |\bar\Ga| = n} |K_{\Ga}| 
\;\leq\;  c \; \eta^n.
\label{convv}
\end{equation}
\end{proposition}
The function $D(i,j)$ will be specified in Proposition 
\ref{pclustest} below. 

\medskip
\begin{proof}
We put for a shorthand  
$\D(\ga) = \prod_{(\tau_i,\tau_j) \in \ga} \e D(i,j)$.
Consider the function of complex variable $z$
\begin{equation}
H(z;\eps) = \sum_{\Ga\in\K_N: \Ga^* \ni 0} K_{\Ga} z^{|\bar \Ga|} = 
\sum_{\Ga\in\K_N: \Ga^* \ni 0 \atop \Ga \supset\; \mbox{\tiny{one contour}}} 
K_{\Ga} z^{|\bar\Ga|} + \sum_{\Ga\in\K_N: \Ga^* \ni 0 \atop \Ga \supset\;
\mbox{\tiny{more than one contour}}}K_{\Ga} z^{|\bar\Ga|}
\label{H}
\end{equation}
We show that for sufficiently small $\e > 0$ this is an analytic 
function of $z$ in a circle of radius $R(\e)$ which diverges as 
$\e \to 0$. Moreover, we show that within this circle $H(z;\e)$ 
is uniformly bounded in $\e$. This will then imply (\ref{convv}) by 
choosing $\eta(\e) = 1/R(\e)$. 

We start by estimating the second sum; the first is simpler as it
involves clusters having a single contour. Our strategy is first to
bound it by a sum taken over graphs whose vertices are the contours of
the same cluster. The sums over graphs will then be bounded by sums
taken over trees.

\medskip
\noindent
\emph{\textbf{Bounds by sums over graphs}} \quad
For each $r \geq 2$ consider in $\K_N$ those clusters $\Ga$ that have $r$ 
contours. For given $\Ga \in \K_N$ let $\V_r = \{\ga_1,...,\ga_r\}$ be the 
collection of these contours. We construct connected graphs $G$ by drawing 
edges between the elements of $\V_r$ considered as vertex set. 
Connected graphs are those for which either $\ga_i^* \cap \ga_j^* \neq \emptyset$ 
or there exists $\rh_l \in \Ga$ such that $\rh^*_l \cap \ga_i^* \neq \emptyset 
\neq \rh_l^* \cap \ga_j^*$. Let $\G_r$ denote the set of all possible such graphs. 

Consider the collection of {\it reduced} chains
$\{\hat\rh_1,...,\hat\rh_s\}$ with the properties:
\begin{enumerate}
\item[(1)]
for every pair $\{\ga_i,\ga_j\} \in G \in \G_r$, $\ga_i^* \cap \ga_j^* = \emptyset$, 
there is at least one chain $\hat\rh_l$ of this collection connecting $\ga_i$ and 
$\ga_j$ (i.e. $\overline{\hat \rh_l} \cap \bar \ga_i = \emptyset = \overline{\hat 
\rh_l} \cap \bar \ga_j$ and $\hat \rh_l^* \cap \ga_i^* \neq \emptyset \neq \hat 
\rh_l^* \cap \ga_j^*$), and for any pair $\{\ga_i,\ga_j\} \not \in G$ such a chain 
does not occur;
\item[(2)]
$\cup_i \bar\ga_i \cap \cup_j \overline{\hat\rh}_j = \emptyset$  and 
$\cup_j \{\partial^-{{\hat\rh}_j},\partial^+ {{\hat\rh}_j}\}  \subset 
\cup_i \ga_i^*$; 
\item[(3)]
$0 \in (\cup_i \ga_i^*) \cup (\cup_j {\hat\rh}_j^*)$;
\item[(4)]
each chain $\hat\rh_k$ connects a pair $\{\ga_i,\ga_j\} \in G \in
\G_r$ 
or fills a gap within a contour $\ga_i$.
\end{enumerate}
We call a collection of reduced chains compatible with graph $G$ if it
satisfies the conditions above and denote it 
$\{\hat\rh_1,...,\hat\rh_s\}_G$. Note that each $\hat\rh$ can join
only one pair of contours. A collection of reduced chains is then
constructed through the following steps:
\begin{enumerate}
\item[(1)]
first remove all chains $\rh_l \in \Ga$ for which $\bar\rh_l \subset \cup_{i=1}^r 
\bar\ga_i$;
\item[(2)]
for all remaining chains $\rh_k \in \Ga$, $\bar \rh_k \not\subset \cup_{i=1}^r 
\bar\ga_i$ remove all intervals from the set $\bar\rh_k \cap (\cup_{i=1}^r \bar\ga_i)$;
\item[(3)]
of the remaining intervals form all possible collections of non-empty chains denoted 
by $\{\hat\rh_1,...,\hat\rh_s\}$.
\end{enumerate} 
Then by Proposition \ref{clustest} we write
\begin{eqnarray}
\label{onecont}
\lefteqn{
\big|\sum_{\Ga: \Ga^* \ni 0 \atop \Ga \supset\;
\mbox{\tiny{more than one contour}}} K_\Ga z^{|\bar\Ga|}\big| 
\; \leq \label{est} } \\ &&
\sum_{r\geq2} \sum_{\{\ga_1,...,\ga_r\}} \sum_{G \in \G_r} 
\sum_{s \geq 0} \prod_{i=1}^r ((|z|(1+\e|z|))^{|\bar\ga_i|} 
\D(\ga_i) \sum_{\{\hat\rh_1,...,\hat\rh_s\}_G \atop 
\mbox{\tiny {$0 \in \cup_i \gamma_i^* \cup {\hat\rh}_i^*$}}} 
\prod_{j=1}^s(\e|z|)^{|\overline {\hat\rh_j}|}. \nonumber
\end{eqnarray}
Note that for fixed $\{\ga_1,...,\ga_r\}$ the collection
$\{\hat\rh_1,...,\hat\rh_m\}$ can be obtained from many possible
collections of chains $\{\rh_1,...,\rh_s\}$. This gives the factor 
$(1 + \e |z|)^{|\bar\ga_i|}$ appearing at the right hand side of 
(\ref{est}). From now on we assume that $\e |z|\le 1$ so that we 
can estimate this factor by $2^{|\bar\ga_i|}$.

Now consider the last sum above involving the reduced chains.
In this sum, either $0$ belongs to a contour or some chains. In
the second case there is a factor of $(\e|z|)^{\dist(0,\{\gamma\})}$
appearing in the sum, so we can estimate the sum by
\begin{equation*}
  \begin{split}
&  (\e|z|)^{\dist(0,{\gamma})/2} \sum_{\{\hat\rh_1,...,\hat\rh_s\}_G 
\atop \mbox{\tiny {$0 \in \cup_i \gamma_i^* \cup {\hat\rh}_i^*$}}}
\prod_{j=1}^s (\e|z|)^{|\overline {\hat\rh_j}|/2}
 \le     
 \sum_k (\e|z|)^{\dist(0,\gamma_k)/2}
 \sum_{\{\hat\rh_1,...,\hat\rh_s\}_G 
\atop \mbox{\tiny {$0 \in \cup_i \gamma_i^* \cup {\hat\rh}_i^*$}}}
\prod_{j=1}^s (\e|z|)^{|\overline {\hat\rh_j}|/2}.
  \end{split}
\end{equation*}
For each contour $\gamma_i$ the sum over all reduced chains belonging
to this contour cannot be larger than $2^{|\overline{\gamma_i}|}$
since the number of such chains is bounded by $\overline{\gamma_i}$
(when every chain separates each two successive intervals in the
contour). Moreover for each couple of contours $(\gamma_i,\gamma_j)$
the contribution to the sum of the chains connecting them is given by
$$
2^{|\overline{\gamma_i}|+|\overline{\gamma_j}|} 
(\e |z|)^{\dist(\gamma_i,\gamma_j)/2}
$$
since there is at least one chain longer than
$\dist(\gamma_i,\gamma_j)$ and the rest of the chains contribute
into the combinatorial prefactor. This gives
\begin{eqnarray}
\label{morecont-2}
\lefteqn{
\big|\sum_{\Ga: \Ga^* \ni 0 \atop \Ga \supset\;
\mbox{\tiny{more than one contour}}}
K_\Ga z^{|\bar\Ga|}\big| \; \leq \label{estt} } \\ &&
\sum_{r\geq2} \sum_{\{\ga_1,...,\ga_r\}}
 \sum_k  \prod_{i=1}^r (8|z|)^{|\bar\ga_i|} \D(\ga_i) 
(\e|z|)^{\dist(0,\gamma_k)/2}  \sum_{G \in \G_r}
 \prod_{\{\gamma_i,\gamma_j\}\in G} 
(\e |z|)^{\dist(\gamma_i,\gamma_j)/2}.
 \nonumber
\end{eqnarray}

\medskip
\noindent
\emph{\textbf{Bounds by sums over trees}} \quad
We use the tree-graph bound  (cf. Lemma 8, Ch. 2, Sect. 4 
of \cite{MM}) to get
\begin{equation}
\sum_{G\in \G_r} \prod_{\{\ga_i,\ga_j\} \in G} (\e |z|)^
{\dist (\gamma_i,\gamma_j)/2} 
\leq 
\prod_{i,j =1}^r (1 + (\e |z|)^{\dist(\gamma_i,\gamma_j)/2}) 
\sum_{T \in \T_r} \prod_{\{\ga_i,\ga_j\} \in T} 
(\e |z|)^{\dist(\gamma_i,\gamma_j)/2}
\label{mm8}
\end{equation}
where $\T_r$ is the set of trees on the vertex set
$\{\gamma_1,\dots,\gamma_r\}$. Moreover we have
$$
\prod_{i,j =1}^r (1 + (\e |z|)^{\dist(\gamma_i,\gamma_j)/2}) 
\le
2^{2 r} \le 2^{2\sum_i |\overline\gamma_i| },
$$
thus (\ref{onecont}) is further estimated by
\begin{eqnarray}
\label{morecont-3}
\lefteqn{
\big|\sum_{\Ga: \Ga^* \ni 0 \atop \Ga \supset\;
\mbox{\tiny{more than one contour}}}
K_\Ga z^{|\bar\Ga|}\big| \; \leq \label{esttt} } \\ &&
\sum_{r\geq2} \sum_{\{\ga_1,...,\ga_r\}}
 \sum_k  \prod_{i=1}^r (32|z|)^{|\bar\ga_i|} \D(\ga_i) 
(\e|z|)^{\dist(0,\gamma_k)/2}  \sum_{T \in \T_r}
 \prod_{\{\gamma_i,\gamma_j\}\in G} (\e |z|)^{\dist
(\gamma_i,\gamma_j)/2}.
\nonumber
\end{eqnarray}

Take the trees over vertex set $\{1,...,r\}$ obtained through 
$\ga_k \mapsto k$, $\forall k = 1,...,r$; denote them $\tilde T$ 
and the set of all such trees by $\tilde \T_r$. Then we re-sum in
(\ref{morecont-3}):
\begin{equation}
  \begin{split}
& \mbox{r.h.s.}\; (\ref{morecont-3}) \; \leq\\
& \quad  \sum_{r=2}^\infty \frac{1}{r!} 
\sum_{\tilde T \in \tilde \T_r} 
\sum_{i^*=1}^r 
\sum_{(\ga_1,...,\ga_r)} \prod_{i=1}^r (32|z|)^{|\bar\ga_i|} 
(\e |z|)^{\dist(0,\ga_{i^*})/2} \D(\ga_i) 
\prod_{\{i,j\} \in \tilde T}
(\e |z|)^{\dist(\ga_i,\ga_j)/2}.    
  \end{split}
\label{hae}
\end{equation}
The change of bracket indicates that the third sum here is performed 
over all ordered collections of disjoint contours. Fix an enumeration
of $\tilde T$ and pick its first element $i_0$. We estimate first 
\begin{equation}
\sum_{(\ga_1,...,\ga_r)} \prod_{i=1}^r (32|z|)^{|\bar\ga_i|} \D(\ga_i) 
\prod_{\{i,j\} \in \tilde T} (\e |z|)^{\dist(\ga_i,\ga_j)/2}.
\label{hva}
\end{equation}
Let $j_0 \neq i_0$ be an end vertex of tree $\tilde T$ being joint 
only with vertex $k_0$. Then
\begin{eqnarray}
\lefteqn{
\hspace{-1.5cm}
\sum_{\ga_{j_0}} (32|z|)^{|\bar\ga_{j_0}|} \D(\ga_{j_0}) (\e |z|)^
{\dist(\ga_{k_0},\ga_{j_0})/2}  \label{i0}} \\ &&
\;\leq\;  \sum_{\tau'' \in \bar\ga_{k_0}} \sum_{\ga_{j_0}} 
\sum_{\tau' \in \bar \ga_{j_0}} (\e |z|)^{\dist(\tau',\tau'')/2} 
(32|z|)^{|\bar\ga_{j_0}|} \D(\ga_{j_0}) 
\nonumber  \\ &&
\;\leq\; \sum_{\tau'' \in \bar\ga_{k_0}} \sum_{\tau'} 
(\e |z|)^{\dist(\tau',\tau'')/2} \sum_{\ga_{j_0}: \tau' \in \bar 
\ga_{j_0}} (32|z|)^{|\bar\ga_{j_0}|} \D(\ga_{j_0}). \nonumber 
\end{eqnarray}
Here we used that $(\e|z|)^{\dist(\ga,\ga')/2} 
\leq 
\sum_{\tau \in \ga, \tau' \in \ga'} (\e|z|)^{\dist(\tau,\tau')/2}$.
By using Lemma \ref{lsumd} below and the bound
\begin{equation}
\sum_{k=2}^\infty a^{k-1} k^m \leq \frac{2^{m}m! e a}{1 - e a}
\label{geo}
\end{equation}
obtained via complex integration, we further estimate (\ref{i0}) by
\begin{equation}
\label{j0}
\sum_{\tau'' \in \bar\ga_{k_0}} \sum_{\tau'} 
(\e |z|)^{\dist(\tau',\tau'')/2} 
\sum_{k=2}^\infty (32|z|)^k (C\e)^{k-1} 
\leq  
\frac{64 C |\bar\ga_{k_0}| \e|z|}{(1 - (\e|z|)^{1/2})
(1 - 32C \e|z|)}. 
\end{equation}
From now on we choose $z$ such that $32C \e|z| < 1$ holds.

Next we go on by taking the next vertex of $\tilde\T$ in line, say
$j_1 \neq i_0$ connecting with $k_1$. We iterate the procedure for the
new tree obtained by deleting from $\tilde T$ the vertex $j_0$ and
edge $(j_0, k_0)$. If $j_1 \neq k_0$, we get again an estimate of the
type (\ref{j0}). If $j_1 = k_0$, we estimate
\begin{equation*}
\sum_{\tau'' \in \ga_{k_1}} \sum_{\tau'} (\e|z|)^{\dist(\tau',\tau'')/2} 
\sum_{\ga_{j_1}: \tau' \in \ga}|\bar\ga_{j_1}|(32|z|)^{|\bar\ga_{j_1}|}
\D(\ga)
\leq  
\const{q2} |\bar\ga_{k_1}| \sum_{k=2}^\infty (32|z|)^k 
k (C \e)^{k-1},
\end{equation*}
with some $\defconst{q2} > 0$. Continuing this procedure inductively 
we get after summation over $\ga_{j_m}$, $j_m \neq i_0$, connected 
to $\ga_{k_m}$, the net contribution
\begin{eqnarray*}
{\const{q2}|\bar\ga_{k_m}|} \sum_{k=1}^\infty (32|z|)^k k^{l_{j_m} -1} 
(C\e)^{k-1} 
&\leq&
{\const{q2}|\ga_{k_m}|} \sum_{k=2}^\infty (\const{q3}|z|)^k 
k^{l_{j_m}-1} (C\e)^{k-1} \\
&\leq&
\const{q9} \; |\ga_{k_m}||z| \sum_{k=2}^\infty ({\const{q5}} 
\e |z|)^{k-1} k^{l_{j_1}-1}
\label{ind}
\end{eqnarray*}
where $l_{j_m}$ is the degree of vertex $j_m$, i.e. the number 
of edges of $\tilde T$ incident to $j_m$, and $\defconst{q3},
\defconst{q4} > 0$, $\defconst{q5} = C\const{q3}$, $\defconst{q9} 
= \const{q3}\const{q4}$ is the long sequence of constants. By using 
(\ref{geo}) again, we estimate (\ref{hva}) further for fixed 
$\ga_{i_0}$ and $\tilde T$ to get
\begin{equation*}
\sum_{\ga_k: \; k \neq i_0 \atop k = 1,...,r} \prod_{\{i,j\} \in \tilde T} 
(\e|z|)^{\dist(\tau_i,\tau_j)/2} \prod_{i \neq i_0} \D(\ga_i) (32|z|)^
{|\bar\ga_i|} \leq
|\bar\ga_{i_0}| ^{l_{i_0}} \prod_{k \neq i_0} 2^{l_{j_k}} (l_{j_k}-1)! \; 
(k_3 \e |z|^2)^{r-1}  
\end{equation*}
where we used that $\e |z| \leq 1/(\const{q5} e)$. Thus we need 
furthermore (see (\ref{hae}))
\begin{eqnarray*}
\lefteqn{
\sum_{\ga_{i_0}} |\bar\ga_{i_0}|^{l_{i_0}} \D(\ga_{i_0})
(16|z|(1+\e|z|))^{|\ga_{i_0}|} (\e|z|)^{\dist(0,\ga_{i_0})/2} } \\ &&
\hspace{0.3cm} \leq 
\sum_{\tau} (\e|z|)^{\dist(0,\tau)/2} \sum_{\ga_{i_0}: \tau \in \ga_{i_0}} 
|\bar\ga_{i_0}|^{l_{i_0}} \D(\ga_{i_0}) (32|z|)^{|\ga_{i_0}|}.
\end{eqnarray*}
By a repetition of the arguments above we get 
\begin{eqnarray*}
\sum_{\ga_{i_0}: \tau \in \ga_{i_0}} |\bar\ga_{i_0}|^{l_{i_0}}
\D(\ga_{i_0}) (32|z|)^{|\ga_{i_0}|}
& \leq &
\sum_{k=2}^\infty k^{l_{i_0}} (C\e)^{k-1} (32|z|)^k \\
& \leq &
\const{q6} 2^{l_{i_0}} l_{i_0}! \e |z|^2,
\end{eqnarray*}
with $\defconst{q6} > 0$. Summation over $\tau$ gives
$
\sum_\tau (\e|z|)^{\dist(0,\tau)/2} \leq \; \const{q7}, 
$
with some $\defconst{q7} > 0$, hence we finally obtain for fixed 
$\tilde T$ and $i_0$
\begin{eqnarray*}
2^{l_{i_0}} l_{i_0}! \prod_{j_k \neq i_0} 2^{l_k} (l_{j_k}-1)! 
\const{q7} (\const{q6} \e |z|^2)^{r-1}
&\leq&
\const{q7} (2^2 \const{q6} \e |z|^2)^{r-1} \prod_{k=1}^r l_{j_k}! 
\end{eqnarray*}
where we used the fact $\sum_{k=1}^r l_{j_k} = 2(r-1)$ for trees.  
An upper bound on the number of trees with vertices $\{1,...,r\}$ 
and degrees $\{l_1,...,l_r\}$ is \cite{MM}
\begin{equation}
\frac{2^{r-2} (r-2)!}{\prod_{j=1}^r l_j!}. 
\label{inc}
\end{equation}
Moreover, the number of collections $\{l_1,...,l_r\}$ such that $l_i > 
0$ and $\sum_i l_i = 2(r-1)$ is bounded from above by $2^{2(r-1)}$.
Hence, by summing over $i_0$ and combining this estimate with
(\ref{inc}), we get
\begin{equation}
\sum_{\Ga: \Ga^* \ni 0 \atop \Ga \supset\;\mbox{\tiny{more than one
      contour}}} K_\Ga |z|^{|\bar\Ga|} \leq c \e |z|^2
\end{equation}
with some constant $c > 0$. This completes the estimate of the second
term in (\ref{H}). The first term there can be handled in a similar
way with substantial simplifications due to the fact that only one 
contour occurs in the clusters. 

It is seen then that by choosing $z$ such that $\e |z|^2 \leq
\mbox{const}$, the sum $\sum_\Ga K_{\Ga} z^{|\bar \Ga|}$ converges
and is bounded. Hence $H(z)$ is an analytic function within a circle 
of radius $R(\e)$ with $R(\e) \to \infty$ as $\e \to 0$, and is 
bounded by a constant independent of $\e$. Thus 
\begin{equation}
\sum_{\Ga: \Ga^* \ni 0 \atop |\bar\Ga| = n} |K_{\Ga}| \;\leq\; 
\mbox{const} \; R(\e)^{-n} \; := \;  \mbox{const} \; \eta(\e) ^n,
\end{equation}
with suitable constants.
\end{proof}

Finally we show the lemma referred to in the proof above. 
\begin{lemma}
\label{lsumd}
There is a constant $C > 0$ such that for any interval $\tau$ 
and integer $k \geq 2$ 
\begin{equation}
\sum_{\ga: \bar\ga \ni \tau \atop |\bar\ga|=k} \D(\ga) \; \leq\; 
(C\e)^{k-1}.
\label{sumd}
\end{equation}
\end{lemma}
\begin{proof}
\begin{equation}
\sum_{\ga: \tau \in \bar\ga \atop |\bar\ga| = k} \D(\ga) = 
\sum_{\{\tau_1,...,\tau_k\}} \sum_{G\in \hat\G}
\prod_{\{\tau_i,\tau_j\} \in G} \e D(i,j)
\end{equation}
Here $\hat\G$ denotes the set of connected graphs with vertices 
$\tau_1,...,\tau_k$. Note that for fixed $i_0$ we have 
$\sum_{\tau_j: j \neq i_0}  \e D(i_0,j) \le C \e$ with some 
$C > 0$. Thus by using (\ref{mm8}) we find
\begin{equation}
\sum_{G\in\hat\G} \prod_{\{\tau_i,\tau_j\} \in G}  \e D(i,j)
\leq 
e^{k C \e} \sum_{T\in \hat\T} \prod_{\{\tau_i,\tau_j\} \in T} 
\e D(i,j) 
\label{102}
\end{equation}
where $\hat\T$ are trees for the same vertex set as for $\hat\G$. 
Next order the  collection $\{\tau_1,...,\tau_k\}$ further to get 
\begin{equation}
\sum_{\{\tau_1,...,\tau_k\}} \sum_{T \in \hat T} \prod_{\{i,j\} 
\in \tilde \T} \e D(i,j)  =
\frac{1}{k!} \sum_{(\tau_1,...,\tau_k)}\sum_{\tilde T \in\tilde\T} 
\prod_{\{i,j\} \in \tilde T}  \e D(i,j),
\label{101}
\end{equation}
with the same $\tilde \T$ as previously. We then obtain inductively
\begin{equation}
\sum_{(\tau_1,...,\tau_k)} \prod_{(\tau_i,\tau_j) \in \T} \e D(i,j)
\leq (C \e)^{k-1}.
\end{equation}
Since the number of trees having $k$ vertices is $k^{k-2}$ \cite{MM}, 
by using Stirling's formula, (\ref{102}) and (\ref{101}) we complete 
the proof of the lemma.
\end{proof}

\subsection{Cluster estimates}
Having the abstract cluster expansion at hand, we turn now to
establishing the bounds~(\ref{clustest-abs}) on the cluster
activities. 
\begin{proposition}
There exists $\delta > 1$ and a function
$\e(\lambda)<\infty$ with $\e(\lambda) \to 0$ as $\lambda \to 0$ and
$b = b(\lambda)\ge 1$ such that for every $N > 0$ and every cluster
$\Ga = \{\ga_1,...,\ga_r;\rh_1,...,\rh_s\} \in \K_N$, the bound 
\begin{equation}
|K_{\Gamma}| \; \leq \; \prod_{l=1}^r \prod_{(\tau_i,\tau_j) \in 
\gamma_l} \frac{\e}{(1+b|i-j-1|)^\delta }\prod_{m=1}^s  \e^{|\bar\rh_m|}
\label{clustest}
\end{equation}
holds.
\label{pclustest}
\end{proposition}
The first product (over the contours) above is our $D(i,j)$ in
Proposition \ref{p2} above, and it is readily seen that it satisfies
the condition given there.

\medskip
\begin{proof}
By H\"older inequality
\begin{eqnarray}
|K_{\Ga}| 
& \leq & 
\prod_{l=1}^r \prod_{(\tau_i,\tau_j) \in \gamma_l} \left(\int \expect_{\nuxT}
\left[|e^{-\lambda W_{\tau_i,\tau_j}} - 1|^{n_{ij}}\right]
d\omega^{\otimes N}(\mathbf{x}) \right)^{1/n_{ij}} 
\times \\
&& \hspace{0.3cm} \times 
\prod_{m=1}^s \prod_{k: \tau_k \in \rh_m}\left(\int|\pi_b(x_{k+1},x_k) - 1|^\beta 
d\omega(x_k)d\omega(x_{k+1})\right)^{1/\beta} \nonumber
\end{eqnarray}
with suitable exponents. We choose $\beta = 4$, $n_{ij} = A |i-j+1|^{\Delta}$, 
with $\Delta > 1$ to be specified below. Taken with correct 
multiplicities, we pick $A$ such that
$$
\frac{2}{\beta} + \sum_{j \in \mathbb{N} \atop j \geq i}\frac{2}{n_{ij}} =
\frac{1}{2} + \frac{2}{A}\sum_{k=1}^\infty \frac{1}{k^\Delta} \leq 1.
$$
The first part of estimate (\ref{clustest}) follows by Lemma \ref{holderw}, 
the second by Lemma \ref{holderpfi} below. By choosing $b= - \log |\lambda|/
(\Lambda+C)$ with suitable $C>0$, we have
$
\lambda e^{Cb} = e^{-\Lambda b} = |\lambda|^{\Lambda/(\Lambda+C)} =: \e,
$
thus the estimate (\ref{clustest}) is finally obtained.
\end{proof}

\begin{lemma}
\label{lemma:trans}
For large enough $b > 0$ there is a constant $C > 0$ such that
\begin{equation}
|\pi_b(x,y) - 1| \,\leq\, Ce^{-\Lambda b},
\end{equation}
uniformly in $x,y \in \R^d$, where $\Lambda > 0$ is the spectral gap of the 
Schr\"odinger operator $H = H_0 + V$. 
\label{holderpfi}
\end{lemma}
\begin{proof}
By assumption the potential is chosen so that $H$ is intrinsically 
ultracontractive, i.e., for each $b > 0$, $C_b = \norm[L^{\infty}(\R^{2d})]{\pi_b} 
< \infty$. 
 By the semigroup property of $\pi_b$ and the fact that $\int \pi_b(x,y) d\omega(y) = 
 1$ for each $x$, for $b > 2$ we have
 \begin{equation}
 \label{iucestimate}
 \begin{split}
 |\pi_b(x,y) - 1| 
 & =  
 \Big| \int d\xi \int d\eta \,  \pi_1(x,\xi)\Psi^2(\xi) (\pi_{b-2}(\xi,\eta) - 1) 
 \Psi^2(\eta) \pi_1(\eta,y) \Big| \\
 & \leq  
 C_1^2 \int d\xi \int d\eta \, \Psi(\xi) | \hat\pi_{b-2}(\xi,\eta) - 
 \Psi(\xi)\Psi(\eta)| \Psi(\eta)  \\
 & \leq  
 C_1^2 \left( \int d\xi \int d\eta (\hat\pi_{b-2}(\xi,\eta) - \Psi(\xi) \Psi(\eta))^2 
 \right)^{1/2}  \\
 & = 
 C_{1}^2 e^{-(b-2)(E_1-E)} \left(\sum_{k\geq 2} e^{-2(b-2)(E_k- E_1)} \right)^{1/2} 
\leq  
 Ce^{-\Lambda b},
  \end{split}
  \end{equation}
where $E = \inf \Spec H$.
The last but first step comes about as follows. $\norm[L^{\infty}(\R^{2d})]{\pi_b} < \infty$ 
implies that $\hat\pi_b \in L^2(\R^{2d},dx)$ for each $b > 0$. Thus $e^{-bH}$ is a 
Hilbert-Schmidt operator for each $b>0$, in particular $H$ has a purely discrete spectrum 
with eigenvalues $E < E_{1} \leq E_{2} \leq \ldots$. With $P_\Psi$, the projection onto 
the subspace of $L^2(\R^d,dx)$ spanned by $\Psi$, the last equality gives the Hilbert-Schmidt 
norm of $e^{-(b-2)H} - P_\Psi$.
\end{proof}

\medskip
\begin{lemma}
Assume that $\beta > 3$. Then there exists $\delta > 1$ and
constants $\defconst{ggg}, \defconst{ggg2} < \infty$ such that,  
for $\lambda$ small enough
\begin{equation}
\begin{split}
\left(\int \expect_{\nuxT}\left[ |e^{\lambda
W_{\tau_i,\tau_j}}-1|^{n_{ij}}\right] d\omega^{\otimes N}
(\mathbf{x}) \right)^{1/n_{ij}} 
& \le
\const{ggg}  \lambda  e^{C b}  (1+b|i-j-1|)^{\delta}
\end{split}
\end{equation}
for all $i,j$ and $b \ge 1$.
\label{holderw}
\end{lemma}
\begin{proof}
Note that the conditional expectation $\expect_{\nuxT}
[|e^{\lambda W_{\tau_i,\tau_j}}-1|^{n_{ij}}]$ depends only on 
$x_i,x_{i+1},x_{j},x_{j+1}$. Thus we can write with a slight 
abuse of notation,
$$
\int \expect_{\nuxT} 
\left[|e^{\lambda W_{\tau_i,\tau_j}}-1|^{n_{ij}}\right] 
d\omega^{\otimes N}(\mathbf{x}) =
\int \expect_{\nuxT} 
\left[|e^{\lambda W_{\tau_i,\tau_j}}-1|^{n_{ij}}\right] 
d\omega^{\otimes 4}(x_i, x_{i+1}, x_j, x_{j+1})
$$
when $|i-j|>1$, otherwise one integral must be ignored in this 
expression. Then by (\ref{eq:bridge-2}) the problem reduces to 
estimates on the multiple Brownian bridge $\widehat\WxT$:
\begin{equation*}
  \begin{split}
& A = \int \expect_{\nuxT} 
      \left[|e^{\lambda W_{\tau_i,\tau_j}}-1|^{n_{ij}}\right] 
       d\omega^{\otimes 4}(x_i,x_{i+1},x_j,x_{j+1}) \\ 
& \le
\int \frac{\expect_{\widehat\WxT} 
      \left[|e^{\lambda W_{\tau_i,\tau_j}}-1|^{n_{ij}} 
       e^{-V_{\tau_i}-V_{\tau_j}}\right]}
       {Z_b \Psi(x_i) \Psi(x_{i+1}) \pi_{b}(x_i,x_{i+1})}
     \frac{\Pi_{b}(x_i,x_{i+1}) \Pi_{b}(x_j,x_{j+1})}
       {Z_b \Psi(x_j) \Psi(x_{j+1}) \pi_{b}(x_j,x_{j+1})} 
       d\omega^{\otimes 4}(x_i,x_{i+1},x_j,x_{j+1}) \\ 
&  \le 
C Z_b^{-2} \int \expect_{\widehat\WxT} 
      \left[|e^{\lambda W_{\tau_i,\tau_j}}-1|^{n_{ij}} 
       e^{-V_{\tau_i}-V_{\tau_j}}\right]
       \Pi_{b}(x_i,x_{i+1}) \Pi_{b}(x_j,x_{j+1}) \times \\ 
& \qquad \times 
\Psi(x_i) \Psi(x_{i+1}) \Psi(x_j) \Psi(x_{j+1})
dx_i dx_{i+1} dx_{j} dx_{j+1},
  \end{split}
\end{equation*}
where we used Lemma \ref{holderpfi} and chose $b$ large enough 
so that
$
\sup_{x,y} |\pi_{b}(x,y)-1|\le C e^{-\Lambda b} \le 1/2
$.
For $i+1<j$ notice that $\Psi(x_{i+1})\Psi(x_{j+1}) \le C$. Then 
by integrating with respect to $x_{i+1},x_{j+1}$ we remove the 
conditional expectation and obtain
\begin{equation}
\label{eq:W-est-1}
  \begin{split}
A
& \le C 
\int \widehat\expect_{x_i,x_j} 
 \left[|e^{\lambda W_{\tau_i,\tau_j}}-1|^{n_{ij}} 
   e^{-V_{\tau_i}-V_{\tau_j}}\right]
   \Psi(x_i) dx_i \Psi(x_j) dx_j, 
 \end{split}
\end{equation}
where now $\widehat\expect_{x_i,x_j}$ denotes expectation over the two
pieces $\XX_{\tau_i}$ and $\XX_{\tau_j}$ weighted by $\W^{x_i}_{\tau_i} 
\otimes \W^{x_j}_{\tau_j}$, i.e. two independent Wiener measures 
starting at $x_i$ and $x_j$, respectively. 

Next we estimate the expectation in (\ref{eq:W-est-1}),
\begin{equation*}
  \begin{split}
&  \widehat\expect_{x_i,x_j} \left[|e^{\lambda
      W_{\tau_i,\tau_j}}-1|^{n} e^{-V_{\tau_i}-V_{\tau_j}}\right]
 \le      \widehat\expect_{x_i,x_j} \left[|\lambda
      W_{\tau_i,\tau_j}|^n e^{n |\lambda
      W_{\tau_i,\tau_j}|} e^{-V_{\tau_i}-V_{\tau_j}} \right] \\ 
& \qquad
  \le    |\lambda|^n  \left(\widehat\expect_{x_i,x_j} \left[|
      W_{\tau_i,\tau_j}|^{2n}\right]\right)^{1/2} \left(\widehat\expect_{x_i,x_j} \left[ e^{4 n |\lambda
      W_{\tau_i,\tau_j}|}\right]\right)^{1/4}  \left(\widehat\expect_{x_i,x_j} \left[e^{-4 V_{\tau_i}-4 V_{\tau_j}} \right]\right)^{1/4}.
  \end{split}
\end{equation*}
Since $V$ is of Kato-class, we have the uniform bound
$$
 \left(\widehat\expect_{x_i,x_j} \left[e^{-4 V_{\tau_i}-4 V_{\tau_j}}
 \right]\right)^{1/4} \le \sup_x  \left(\expect_{x} \left[e^{-4
 V_{[0,b]}} \right]\right)^{1/2} \le C e^{C b},
$$
with some $C > 0$. Furthermore, write
$
 W_{\tau_i,\tau_j} = \int_{\tau_i} dX_t \int_{\tau_j} dX_s g_{ts}
$
(see (\ref{eq:J-rewriting}) below) and estimate the double integral below
by using the Burkholder-Davis-Gundy inequality,
\begin{equation*}
  \begin{split}
\widehat\expect_{x_i,x_j} \left[|
      W_{\tau_i,\tau_j}|^{2n}\right] & \le c_{n} \widehat\expect_{x_i,x_j} \left[\left|
      \int_{\tau_i} dt\left | \int_{\tau_j} dX_s g_{ts}\right|^2
      \right|^{n}\right]  
\\  & \le c_{n} |\tau_i|^{n-1} \int_{\tau_i} dt \widehat\expect_{x_i,x_j} \left[
     \left | \int_{\tau_j} dX_s g_{ts}\right|^{2n} \right]      
\\  & \le c_{n}^2 |\tau_i|^{n-1} |\tau_j|^{n-1} \int_{\tau_i} dt  \int_{\tau_j} ds \widehat\expect_{x_i,x_j} \left[
     |g_{ts}|^{2n} \right]      
\\  & \le c_{n}^2 b^{2n} C^n (1+b|i-j-1|)^{-2 n \beta}.
  \end{split}
\end{equation*}
Now we estimate also the exponential of the energy. By $e^{|x|} \le 2
\cosh x$ we get
 \begin{equation*}
  \begin{split}
 \widehat\expect_{x_i,x_j} \left[ e^{4 n |\lambda
      W_{\tau_i,\tau_j}|} \right] \le  \widehat\expect_{x_i,x_j} \left[ e^{4 n \lambda
      W_{\tau_i,\tau_j}}  \right]    +    \widehat\expect_{x_i,x_j}
      \left[ e^{- 4 n \lambda
      W_{\tau_i,\tau_j}}  \right].   
  \end{split}
\end{equation*}
Each of the expectations in the right hand side can be similarly 
estimated by using Lemma~\ref{lemma:final-exp-bound} below:
\begin{equation*}
  \begin{split}
\widehat\expect_{x_i,x_j} \left[ e^{4 n |\lambda
      W_{\tau_i,\tau_j}|} \right] \le  (1-n^2 a^2)^{-\const{b}} \le C
  \end{split}
\end{equation*}
with 
$
a = \const{a}
    \lambda b   (1+ b |i-j-1|)^{-\beta} 
$ and $na \le 1/2$.
Hence, by making use these inequalities (uniform in $x_i,x_j$), 
and the fact that $\Psi \in L_1(\R^d,dx)$, we arrive at
\begin{equation}
 \label{eq:ewbound}
\begin{split}
& 
\left(\int \widehat\expect_{x_i,x_j}  
\left[ |e^{\lambda W_{\tau_i,\tau_j}}-1|^{n_{ij}}
 e^{-V_{\tau_i}-V_{\tau_j}}\right]  
\Psi(x_i) dx_i  \Psi(x_j) dx_j \right)^{1/n_{ij}} \\
&
\hspace{1cm} \le 
C c_{n_{ij}}^{1/n_{ij}}\lambda be^{C b}(1+b|i-j-1|)^{-\beta}. 
\end{split}
\end{equation}
By using the estimate $c_k \le (2k)^{2k}$ for the universal 
constant in the Burkholder-Davis-Gundy inequality, we furthermore
obtain
$$
 c_{n_{ij}}^{1/n_{ij}}  (1+b|i-j-1|)^{-\beta} \le (2n_{ij})^2
 (1+b|i-j-1|)^{-\beta} .
$$
Recall that $n_{ij} = A |i-j|^\Delta$ with $\Delta >1$ and choose 
$\Delta <\beta$ so that $an_{ij} \to 0 $ as $|i-j|\to\infty$ and 
condition $an_{ij}\le 1/2$ is satisfied uniformly in $i,j$ for 
$\lambda$ small enough. Thus there is a constant $\defconst{fff} > 
0$ such that
\begin{equation}
\begin{split}
\text{r.h.s. (\ref{eq:ewbound})}
\, \le \,
\const{fff}  \lambda b^{2-2\Delta} e^{C b}(1+b|i-j-1|)^{2\Delta-\beta}
\, \le \, 
\const{fff}  \lambda  e^{C b}(1+b|i-j-1|)^{-\delta}, 
\end{split}
\end{equation}
for all $|i-j| > 1$ and $b \ge 1$, where $\delta = \beta-2\Delta>1$
(this is possible by choosing $\Delta>1$ but small enough). For the 
cases $|i-j|=1$ we can follow a similar strategy to estimate
$$
\left(\expect_{\nuxT} \left[|e^{\lambda W_{\tau_i,\tau_j}}-1|^{2n_{ij}} 
e^{-V_{\tau_i}-V_{\tau_j}}\right]\right)^{1/2n_{ij}} 
\le 
\defconst{tt2} e^{\defconst{tt1} b},  
$$ 
where the constants do not depend on either $b$ or $n_{ij}$.
\end{proof}

\subsection{Energy estimates}
Here we estimate $\widehat \expect_{x,y} [e^{\lambda W_{\tau_i,\tau_j}}
]$.
Below we will prove two lemmas that give the basic estimates by
making use of the following result to control exponential
integrability of stochastic integrals. The first lemma will be 
often used for controlling exponential moments of stochastic integrals.
\begin{lemma}
\label{lemma:exp-bound}
If $X$ is Brownian motion and $f$ is an $\F$-adapted process, we 
have the bound
\begin{equation}
  \label{eq:exp-bound}
\expect [e^{\int_0^b f_s dX_s}] \le 
\left\{\expect [e^{2 \int_0^b |f_s|^2 ds} ] \right\}^{1/2}.
\end{equation}
\end{lemma}
\begin{proof}
The proof is a combination of Cauchy-Schwartz and Girsanov Theorems:
  \begin{equation*}
    \begin{split}
 \expect [e^{\int_0^b f_s dX_s}] &= 
 \expect [e^{\int_0^b f_s dX_s -  \int_0^b |f_s|^2 ds +
  \int_0^b |f_s|^2 ds} ]
\\ & \le  \left\{\expect [e^{2 \int_0^b f_s dX_s - \frac{1}{2} \int_0^b
 |2 f_s|^2 ds}]\right\}^{1/2}\left\{\expect [e^{
 2 \int_0^b |f_s|^2 ds} ]    \right\}^{1/2}
\\ & \le \left\{\expect [e^{
 2 \int_0^b |f_s|^2 ds} ]    \right\}^{1/2}.
    \end{split}
  \end{equation*}
\end{proof}

\medskip
\noindent
\emph{\textbf Estimates for separated intervals} \quad
We turn to estimating exponentials of energy contributions in 
(\ref{sum}) by starting with pairs of intervals that are not 
adjacent. Let thus $i>j+1$; in this case the exponent has the 
form $W_{\tau_i,\tau_j} = \J_{ij} + \J_{ji}$ with
\begin{equation}
  \label{eq:J-rewriting}
\J_{i,j} =  \langle C^X_{\tau_i},C^X_{\tau_j}\rangle_W = 
\int_{t_i}^{t_i+b} dX_t \int_{t_j}^{t_j+b} d X_s\, W(X_t-X_s,t-s).  
\end{equation}
Note that  under the
measure $\W^{x_i}_{\tau_i} \otimes \W^{x_j}_{\tau_j}$
the two currents  $\XX_{\tau_i}$ and $\XX_{\tau_j}$ are independent 
and the interaction energy can be written as a  double stochastic
It\^o integral and estimated by using tools borrowed from stochastic 
analysis.
\begin{lemma}
\label{lemma:final-exp-bound}
Let $i>j+1$. There exist positive constants $\defconst{a},\defconst{b}$ 
such that whenever
\begin{equation}
  \label{eq:cond-bound-1}
a := \lambda  \const{a}  b (1+ |t_i-t_j-b|)^{-\beta} < 1,  
\end{equation}
we have
\begin{equation*}
  \begin{split}
\widehat \expect_{x,y} [e^{\lambda W_{\tau_i,\tau_j}}] &
\le \left[1-a^2\right]^{-\const{b}}.
  \end{split}
\end{equation*}  
\end{lemma}
\begin{proof}
We have
\begin{equation*}
  \begin{split}
\left(\widehat \expect_{x,y} [e^{\lambda W_{\tau_i,\tau_j}}]\right)^4 &
\le
\left(\widehat \expect_{x,y} [e^{2\lambda \J_{ij}}]\right)^2 
\left(\widehat \expect_{x,y} [e^{2\lambda \J_{ji}}]\right)^2. 
  \end{split}
\end{equation*}
By using (\ref{eq:exp-bound}) we obtain
\begin{equation*}
  \begin{split}
\left(\widehat \expect_{x,y} [e^{\lambda \J_{ij}}]\right)^2 &\le
\widehat \expect_{x,y} \left[\exp\left(2 \lambda^2 \int_{t_i}^{t_i+b} dt
  \left|\int_{t_j}^{t_j+b}  d X_s W(X_t-X_s,t-s)
  \right|^2\right)\right]
\\ &\le \int_{t_i}^{t_i+b} \frac{dt}{b}
\widehat \expect_{x,y} \left [\exp\left(2 b \lambda^2 
  \left|\int_{t_j}^{t_j+b}  d X_s W(X_t-X_s,t-s) \right|^2\right)\right ]
\\ &\le \int_{t_i}^{t_i+b} \frac{dt}{b}
\widehat \expect_{x,y} \expect_G \left[\exp\left(2 \lambda
  G \sqrt{b}\int_{t_j}^{t_j+b}  d X_s W(X_t-X_s,t-s)\right)\right],
  \end{split}
\end{equation*}
with $G$ a normal random variable defined on a new probability space
and $\expect_G$ the related Gaussian expectation. Then, using again 
(\ref{eq:exp-bound}) yields
\begin{equation*}
  \begin{split}
\left(\widehat \expect_{x,y} [e^{\lambda W_{\tau_i,\tau_j}}]\right)^2 &
\le \int_{t_i}^{t_i+b} \frac{dt}{b} 
\left(\widehat \expect_{x,y} \expect_G \left[\exp\left(8 \lambda^2 G^2 
b \int_{t_j}^{t_j+b} ds |W(X_t-X_s,t-s)|^2 \right)\right]\right)^{1/2}.
  \end{split}
\end{equation*}
The assumptions on $W$ give
\begin{equation*}
  \begin{split}
  \int_{t_j}^{t_j+b} ds |W(Y_t-Z_s,t-s)|^2 & \le 
  \int_{t_j}^{t_j+b} ds  \frac{C}{(1+|t-s|)^{2\beta}}
 \le {\const{c}} b (1+|t_i-t_j-b|)^{-2\beta}
  \end{split}
\end{equation*}
for some constant $\defconst{c}$,
where we used that $t \in [t_i,t_i+b]$. Hence
\begin{equation*}
  \begin{split}
\left(\widehat \expect_{x,y} [e^{2\lambda \J_{ij}}]\right)^2 &
\le  \left(
\expect_G \left[\exp\left[32 \lambda^2 \const{c} b^2 
(1+|t_i-t_j-b|)^{-2\beta} G^2  \right]\right]\right)^{1/2}.
  \end{split}
\end{equation*}
The Gaussian integration can be performed explicitly, yielding
\begin{equation*}
  \begin{split}
\widehat \expect_{x,y} [e^{2\lambda \J_{ij}}] &
\le   
 \left(1-32 \lambda^2 \const{c} b^2 (1+|t_i-t_j-b|)^{-2\beta}
    \right)^{-d/8},
  \end{split}
\end{equation*}
 as soon as
$
1-32 \lambda^2 \const{c} b^2(1+ |t_i-t_j-b|)^{-2\beta} > 0.
$
Thus the claim follows.
\end{proof}

\medskip
\noindent
\emph{\textbf Estimates for adjacent intervals} \quad
The interaction energy estimates between adjacent intervals are 
given by

\begin{lemma}
For all $i = 0,...,N-2$ we have
$
\widehat \expect_{x,y} [e^{\lambda W_{\tau_i,\tau_{i+1}}}] \le
\defconst{diag} e^{\defconst{diag2} b } < \infty
$
for sufficiently small $\lambda$.
\end{lemma}
\begin{proof}
By using that $W(x,t)$ is bounded and arguments similar to those 
of the previous lemma, the required exponential integrability of
$W_{\tau_i,\tau_{i+1}}$ easily follows, at least for sufficiently
small $\lambda$.
\end{proof}

\subsection{Properties of the cluster expansion}
We finally show how the convergence of the cluster expansion of 
$Z_T$ and $\log Z_T$ seen in Proposition \ref{p0} imply existence 
of a limit Gibbs measure $\mu$. 

For any subset $A \subset \R$ let
$$
Z^A_T = 1 + \sum_{n\geq 1} 
\sum_{{\Ga_1,...,\Ga_r \in \C_N \atop \Ga^*_i \cap \Ga^*_j 
= \emptyset, \; i \neq j} \atop A \cap (\cup_i\Ga^*_i) 
= \emptyset} \prod_{i=1}^r K_{\Ga_i}
$$
and write $Z_T^\Gamma := Z_T^{\overline \Gamma \cup \Gamma^*}$.
By the cluster expansion we have
$$
\log Z^\Gamma_T = 1 + \sum_{n\geq 1} 
\sum_{{\Ga_1,...,\Ga_r \in \C_N } \atop \Ga^* \cap (\cup_i\Ga^*_i) 
= \emptyset} \phi^T(\Ga_1,...,\Ga_n)\prod_{i=1}^r K_{\Ga_i}.
$$
Moreover we can define the correlation functions for the clusters by
\begin{equation}
  \label{eq:cluster-corr}
f^\Gamma_T = \frac{Z_T^\Gamma}{Z_T} = \exp\left( - \sum_{n\geq 1} 
\sum_{{\Ga_1,...,\Ga_r \in \C_N } \atop \Ga^* \cap (\cup_i\Ga^*_i) 
\neq \emptyset} \phi^T(\Ga_1,...,\Ga_n)\prod_{i=1}^r K_{\Ga_i} \right) 
\end{equation}
and let $f^\Gamma = \lim_{T\to\infty} f^\Gamma_T $ as the limit
exists by the cluster estimates above and general arguments of 
cluster expansion~\cite{MM}. Moreover we have the uniform estimate
\begin{equation}
  \label{eq:cluster-corr-estimate}
|f^\Gamma_T| \le 2^{|\overline \Gamma|}  
\end{equation}
for $\lambda$ small enough (the constant $2$ can actually be
replaced with any number larger than $1$, provided $\lambda$ is
chosen correspondingly small). Then existence of the infinite time 
limit for the measures $\{\mu_T\}_T$  follows  easily and we have
\begin{proposition}
  \label{prop:infinite-measure}
The local limit $\mu = \lim_{T\to\infty} \mu_T$ exists and satisfies 
the equality
\begin{equation}
  \label{eq:limit-eq-cluster}
\mathbb{E}_{\mu} [F_S] =
\mathbb{E}_{\chi} [F_S] f^{S} + \sum_{n\geq 1} 
\sum_{{\Ga_1,...,\Ga_r \in \C \atop \Ga^*_i \cap \Ga^*_j 
= \emptyset, \; i \neq j} \atop i: S \cap \Ga^*_i 
\neq \emptyset}
 \mathbb{E}_{\chi}[F_S \prod_{i=1}^r \kappa^{\Ga_i}] 
f^{\cup \overline \Ga}
\end{equation}
for any bounded, $\mathcal{F}_S$-measurable function $F_S$, where 
$S$ is a finite union of intervals of the partition considered in 
the cluster expansion. Moreover, the measure $\mu$ is invariant 
with respect to time shift.
\end{proposition}
\begin{proof}
We have
\begin{equation}
\mathbb{E}_{\mu}[F_S] = \lim_{T \rightarrow \infty} 
\frac{\int e^{-\lambda W_T(X)} F(X) d\nu_T(X)}
{\int e^{-\lambda W_T(X)} d\nu_T(X)} = 
\lim_{T \rightarrow \infty} \frac{Z_T(F)}{Z_T}.
\end{equation}
By the cluster expansion we are led to 
\begin{eqnarray}
Z_T(F) = 
\mathbb{E}_{\chi} [F_S] Z_T^{S} + \sum_{n\geq 1} 
\sum_{{\Ga_1,...,\Ga_r \in \C_N \atop \Ga^*_i \cap \Ga^*_j 
= \emptyset, \; i \neq j} \atop i: S \cap \Ga^*_i 
\neq \emptyset}
 \mathbb{E}_{\chi}[F_S \prod_{i=1}^r \kappa^{\Ga_i}] 
Z_T^{\cup \overline \Ga}.
\label{ce}
\end{eqnarray}
If $F_S$ is bounded, standard arguments show that the series 
on the right hand side is absolutely convergent uniformly in 
$N$ and thus~(\ref{eq:limit-eq-cluster}) follows. Given the 
uniqueness of the limiting measure, its invariance with respect 
to time shifts is a direct consequence of the invariance of the 
potentials and of the It\^o-measure~(for more details 
see~\cite{LM}).
\end{proof}

\begin{corollary}
\label{cor:equiv-mu-nu}
Let $F\in\mathcal{F}_{[0,b]}$ be a positive random variable. 
Then
$
\mathbb{E}_{\mu}[F] \leq C (\mathbb{E}_{\nu}[F^2])^{1/2}      
$.  
\end{corollary}
\begin{proof}
By using Proposition ~\ref{prop:infinite-measure} for $S=[0,b]$ 
(which for fixed $N$ is the interval $\tau_{N/2}$ of the 
partition) we have
$$
\mathbb{E}_{\mu} [F_S] = \mathbb{E}_{\chi} [F_S]  
f^S + \sum_{\Ga_0: \bar\Ga_0 \cap [0,b] \neq \emptyset} 
\mathbb{E}_{\chi} [F_S \kappa^{\Ga_0}] f^{\Ga_0},   
$$
where in the second term the sum is over the only  cluster which
can overlap with $S$. We have $|f^S| \le 2$, $|f^{\Ga_0}| \le 
2^{|\Ga_0^*|}$. On the other hand, by using Lemma~\ref{lemma:trans} 
and choosing $b$ large enough to ensure that 
$\sup_{x,y}|\pi_b (x,y)-1|\le 1/2$, we have
\begin{equation}
\mathbb{E}_{\chi} [F_S] = \int F_S(X) \frac{d\nu^{x_0,x_1}(X)}
{\pi_b(x_0,x_1)} \pi_b (x_0,x_1) d\omega(x_0) d\omega(x_1)  
\; \leq \;  
C \; \mathbb{E}_{\nu} [F_S]. 
\nonumber
\end{equation}
Furthermore,
$
\mathbb{E}_{\chi^N} [\kappa_a^{\Ga_0}] \leq 
\left(\mathbb{E}_{\chi^N} [\xi_a]\right)^{1/2} \; 
\left(\mathbb{E}_{\chi^N} [(\kappa^{\Ga_0})^2]\right)^{1/2}
$
and by the same arguments as in Proposition~\ref{p2} above we 
obtain the bound
$$
\sum_{\Ga_0: [0,b]\in \bar \Ga_0 } [ \mathbb{E}_{\chi} 
(\kappa^{\Ga_0})^2]^{1/2} \; 2^{|\bar\Ga_0|} \; \leq \; 
\mbox{const}
$$
with some constant.  Hence we get that
$
\mathbb{E}_{\mu}[F_S] \le C
\mathbb{E}_{\nu}[F_S] + C
(\mathbb{E}_{\nu}[F_S^2])^{1/2}      
$,
which implies the claim.
\end{proof}

\begin{theorem}
\label{th:mu-lift}
There exists a unique forward current $\mu^\sharp$ on $\Xi$ 
such that its $\Xi$-marginal is $\mu$. Moreover, under $\mu$ 
we have $\mathcal{N}_\alpha(X) < \infty$ almost surely and 
the boundary currents are well defined under $\mu^\sharp$.
\end{theorem}
\begin{proof}
Corollary~\ref{cor:equiv-mu-nu} implies that the measure $\mu$ 
is absolutely continuous with respect to $\nu$, thus the almost 
sure events of $\nu$ carry over to $\mu$ and we can consider 
the lifted measure $\mu^\sharp$. This further implies that
$$
\expect_\mu[N_{[k,k+1]}(X)^3] 
\le 
\left(\expect_\nu[(N_{[k,k+1]}(X))^6]\right)^{1/2} 
\le 
C, 
$$
independently of $k\in\Z$ and then $\expect_\mu[\mathcal{N}_
\alpha(X)] < \infty$ for any $\alpha > 1$. This last 
condition guarantees the existence of the boundary currents 
under the measure $\mu^\sharp$.
\end{proof}

\section{Properties of the Gibbs measure}

\subsection{Dependence on boundary conditions and DLR uniqueness}

Uniqueness in DLR sense means that for any increasing sequence of real 
numbers $\{T_n\}_n$, $T_n \uparrow \infty$, and any corresponding 
sequence of boundary conditions $\{Y_n\}_n \subset \Xi$ we have 
$\mathbb{E}_{\rho^\sharp_{T_n}(\cdot |\YY_n)}[F_B] \to \mathbb{E}_\mu[F_B]$, 
for every bounded $B \subset \R$, and each bounded and local (i.e., 
measurable with respect to $\F_B$) function $F_B$ on $\Xi$. However,
such a strong statement cannot be made in this context and we have to
restrict the class of allowed boundary conditions to be able to
control the limit. Fix $\alpha > 1$ and let
\begin{equation}
\Xi_a = \{Y \in \Xi : \mathcal{N}_\alpha(Y) \le a \},
\quad \Xi_* = \cup_{a > 0} \Xi_a
\end{equation} 
be the set of \emph{allowed boundary conditions} carrying full 
$\mu^\sharp$ measure. Then we have
\begin{theorem}
\label{th:uniq}
For any $a>0$ the measure $\mu$ is unique in DLR sense for any 
sequence of boundary conditions $(\YY_n)_n$ in $\Xi_a$, i.e.
\begin{equation}
\lim_{n\to\infty} \mathbb{E}_{\rho^\sharp_{T_n}(\cdot|\YY_n)}[F_B] =
\mathbb{E}_\mu[F_B].
\label{dlruniq}
\end{equation}
\end{theorem}
\begin{proof}
We consider the class of bounded local functions $F_S$ on $\Xi$ indexed 
by bounded intervals $S \subset \R$ (that is, $F_S$ is measurable with 
respect to $\F_S$). It suffices to prove that for any increasing sequence 
$\{T_n\}$, $T_n \to \infty$, and any corresponding sequence of boundary 
conditions $(\YY_n)_n \subset \Xi_a$ (\ref{dlruniq}) holds for arbitrary 
$F_S$ of the above class. To show this we express the conditional 
expectations appearing above in terms of the cluster representation. We 
suppose without loss that $S$ consists of a finite union of intervals of 
the partition of $[-T,T]$. 

From now on we follow the steps of the construction of the cluster
representation in Section \ref{sec:cluster-rep}. Take the same partition 
of the interval $[-T,T]$ into disjoint segments as before. The interaction 
energy can then be written as
\begin{equation}
W_T(X|Y) = 
\sum_{0 \leq i < j \leq N} W_{\tau_i,\tau_j}(X_{\tau_i},X_{\tau_j}) + 
\sum_{0 \leq i \leq N-1} W_{\tau_i,T}^Y (X_{\tau_i}),
\end{equation}
with the same notations as before, and with 
\begin{equation}
W_{\tau_i,T}^Y (X_{\tau_i}) = 2 \int_{\tau_i}w^{C^{Y +}_{T}}(X_t,t) dX_t + 
2\int_{\tau_i}  w^{C^{Y-}_{-T}}(X_t,t) dX_t.   
\end{equation}
By (\ref{a22}) the estimate 
\begin{equation}
\label{a22y}
\sup_{x\in \R^d, t \in \tau} |w^{C^{Y \pm}_{\pm T}}(x,t)| 
\leq 
\frac{M_\tau \|C^{Y \pm}_{\pm T}\|_{\mathcal{D}'}}
{\left(\dist(\tau,[-T,T]^c) + 1 \right)^{\beta-\alpha}}
\end{equation}
easily follows. (Here $[-T,T]^c = \R \smallsetminus [-T,T]$ and 
$\dist(\tau_k,[-T,T]^c) = \min \{kb, (N-k-1)b\}$.) 

Fix the positions $X_{t_0} = Y^-_{-T} = x_0$, $X_{t_1} = x_1$, ..., 
$X_{t_{N-1}} = x_{N-1}$, $X_{t_N} = Y^+_T = x_N$. Similarly to 
(\ref{auxi}) introduce the auxiliary measure 
\begin{equation}
d\chi_N^Y = 
\prod_{k=0}^{N-1} \frac{e^{-\lambda W^Y_{\tau_k,T}(X_{\tau_k})}}
{\Z_{\tau_k}^T(Y|x_k,x_{k+1})} d\nu^{x_k,x_{k+1}}_b (X_{\tau_k}) 
\prod_{k=1}^{N-1} d\omega(x_k) 
\end{equation}
where 
\begin{equation}
\Z^T_{\tau_k} (Y|x_k, x_{k+1}) = \mathbb{E}_{\nu^{x_k,x_{k+1}}_b} 
\left[ e^{-\lambda W_{\tau_k,T}^Y (X_{\tau_k})} \right].
\end{equation}
Also, for every cluster $\Ga$ consider the function $\kappa_{\Ga}^Y$ 
defined similarly to (\ref{kappa}). If $\pm T \not \in \Ga^*$, then 
$\kappa^Y_\Ga$ does not depend on $Y$. If $-T \in \Ga^*$ and/or $T \in 
\Ga^*$, then $\kappa^Y_\Ga$ depends on $Y^-_{-T} = x_0$ and/or $Y^+_T = 
x_N$, respectively. Next we define the weights
\begin{equation}
K^Y_{\Ga} = \mathbb{E}_{\chi^Y_N} [\kappa^Y_\Ga]
\end{equation}
in the same manner as in (\ref{aux}). The partition function $Z_T(Y)$ 
can be expressed similarly to (\ref{clexp}) with these altered objects. 
Note that
\begin{enumerate}
\item[(1)]
for sufficiently small $|\lambda| \neq 0$ the cluster estimate 
(\ref{convv}) obtained in Proposition \ref{p2} stays essentially 
valid, i.e., 
\begin{equation}
\label{eq:bound-KY}
\sum_{\Ga: \Ga^* \ni 0 \atop |\bar \Ga| = n} |K^Y_\Ga| \;\leq\; 
c \eta'(\lambda) ^n,
\end{equation}
with $\eta'(\lambda)$ going to zero as $\lambda \to 0$;
\item[(2)]
for any fixed $\Ga$ we have
$
\lim_{T\to \infty} \kappa^Y_\Ga = \kappa_\Ga
$ and 
\begin{equation}
\label{eq:convergence-KY}
\lim_{T\to \infty} K^Y_\Ga = K_\Ga,
\end{equation}
both uniformly convergent in $\YY \in \Xi_a$. 
\end{enumerate}
The proof of these statements goes by the same arguments  used 
in the previous section and it will be omitted. The 
bound~(\ref{eq:bound-KY}) can be proven as in Lemma~\ref{p2}.
Indeed, by using Lemma~\ref{lemma:exp-bound} and the bound~(\ref{a22y}) 
on the influence of the boundary current we have a handle to control 
the exponential moments of  $W_{\tau_i,T}^Y(X_{\tau_i})$ in terms of 
the norm of the boundary current and repeat the proof of 
Lemma~\ref{holderw} to obtain the necessary estimates on cluster 
activities (with constants depending on $a$). A good control of the 
exponential moments is the key to obtain (\ref{eq:convergence-KY}).

Then in (\ref{dlruniq}) we have
\begin{equation}
\mathbb{E}_{\mu_T} [F_S|Y] = \mathbb{E}_{\chi^Y_T} \left[F_S \right] 
\; f^S_T(Y) + \sum_{n\geq 1} \sum_{\{\Ga_1,...,\Ga_m\}: 
\Ga^*_i \cap \Ga_j^* \neq \emptyset \atop \Ga_i^* \cap S \neq 0, 
\Ga_i \subset [-T,T], i = 1,...,m} 
\mathbb{E}_{\chi^Y_T} \left [F_S \prod_{i=1}^m \kappa^Y_T \right] \; 
f_T^{\cup \bar\Ga}(Y)
\label{eee}
\end{equation}
with the same notations as in (\ref{ce}) and $f^A_T(Y)= Z_T^A(Y)/Z_T(Y)$.

Take now a collection of intervals $\{\tau_i\} = \U$; the partition 
function $ Z_T^{\U}(Y) := Z_{T}^{\cup_{\tau_i \in \U} \tau^*_i} (Y)$ 
can then be written like in~(\ref{clexp}) except for changing $K_\Ga$ 
for $K^Y_\Ga$.
\begin{lemma}
For sufficiently small $|\lambda| \neq 0$ we have the following 
properties of $f^{\U}_T(Y) := Z_T^\U (Y)/Z_T(Y)$. On the one hand,
\begin{equation}
|f^{\U}_T(Y)| \;\leq\; 2^{|\U|},
\label{zup}
\end{equation}
with $|\U|$ denoting the number of intervals contained in $\U$. 
On the other hand,
\begin{equation}
\lim_{T \to \infty} f^{\U}_T(Y) = f^{\U}
\label{zap}
\end{equation}
uniformly in $\YY \in \Xi_a$. Moreover, $f^{\U}$ also satisfies 
(\ref{zup}) above. 
\end{lemma}
\begin{proof}
Both statements are direct consequences of the 
bounds (\ref{eq:bound-KY}) and of (\ref{eq:convergence-KY})
together with the cluster representation of the correlation 
functions~(\ref{eq:cluster-corr}). By putting
\begin{equation}
\hat K_\Ga^Y = \left\{ 
\begin{array}{ll}
K_\Ga^Y & \mbox{if $\Ga$ lies inside $[-T,T]$} \nonumber \\
0 & \mbox{otherwise} 
\end{array} \right.
\end{equation}
and using dominated convergence we obtain that $f_{\U}^Y \to 
f_{\U}$ uniformly as $T \to \infty$.
\end{proof}

We now return to the expression (\ref{eee}). By ergodicity 
of the reference measure  
\begin{equation}
\lim_{T\to \infty} \mathbb{E}_{\chi_N^Y} [F_S] = 
\mathbb{E}_{\chi} [F_S],  
\end{equation}
and hence the first term of (\ref{eee}) converges to 
$\mathbb{E}_{\chi}[F_S] f_S$. By the same argument as above 
we also obtain
\begin{equation}
\lim_{N \to \infty} \mathbb{E}_{\chi_N^Y} 
\left[F_S \prod_{i=1}^m \kappa_{\Ga_i}^Y \right] 
= 
\mathbb{E}_{\chi}\left[F_S \prod_{i=1}^m \kappa_{\Ga_i}\right] 
\label{ele}
\end{equation}
uniformly in $Y$, and 
\begin{equation*}
|\mathbb{E}_{\chi_N^Y} 
[F_S \prod_{i=1}^m \kappa_{\Ga_i}^Y ]| \;\leq\; (\sup|F_S|) 
\prod_{i=1}^m |K^Y_{\Ga_i}|,
\label{maa}
\end{equation*}
which implies
\begin{eqnarray*}
\sum_{\{\Ga_1,...,\Ga_m\} \atop \Ga^*_j \cap S^* = 
\emptyset, j = 1,...,m} 
\left ( \prod_{i=1}^m  |K^Y_{\Ga_i}|\right)  
2^{|\cup_i \bar\Ga_i|} 2^{|S^*|}
& \leq &
2^{|S^*|} \sum_{m=1}^{|S^*|} {|S^*| \choose m} 
\left( \sum_{\Ga: \Ga \ni 0} 2^{|\bar\Ga|} |K^Y_\Ga| \right)^m \\
& \leq &
2^{|S^*|} \sum_{m=1}^{|S^*|} {|S^*| \choose m} 
\left( \sum_{n=2}^\infty 2^n c \eta(\lambda)^n \right)^m < \infty.
\end{eqnarray*}
Here $S^*$ denotes the time points occurring in $S$. By using now 
this estimate together with (\ref{ele}) and applying Lebesgue's 
dominated convergence theorem once again, we arrive at
\begin{equation}
\mathbb{E}_{\chi} [F_S] f_S + \sum_{m=1}^\infty 
\sum_{\{\Ga_1,...,\Ga_m\} \atop \Ga_i^* \cap S^* \neq
\emptyset, i = 1,...,m} \mathbb{E}_{\chi} 
\left[F_S \prod_{i=1}^m \kappa_{\Ga_i} \right] 
f_{S \cup (\cup_i \bar\Ga_i)} = \mathbb{E}_\mu [F_S].
\end{equation}
\end{proof}

\begin{corollary}
The measure $\mu^\sharp$ satisfies the DLR equations
$$
\int_\Xi \expect_{\rho^\sharp_T(\cdot|\YY)}[F_S] 
d\mu^\sharp(\YY) = \expect_{\mu^\sharp}[F_S]
$$  
for any $S<T$.
\end{corollary}
\begin{proof}
Note that $|\expect_{\rho^\sharp_T(\cdot|\YY)}[F_S]| \le \sup |F_S|$ 
and that $\cup_{a>0} \Xi_a$ has full $\mu^\sharp$-measure so that 
the left hand side of the equality makes sense. Fix $a>0$. By 
Theorem~\ref{th:uniq} we have
$$
\int_\Xi \expect_{\rho^\sharp_T(\cdot|\YY)}[F_S] d\mu^\sharp(\YY) 
= 
\lim_{T_2\to+\infty}\int_\Xi \expect_{\rho^\sharp_T(\cdot|\YY)}[F_S] 
d\rho^\sharp_{T_2}(\YY|\ZZ)
$$
uniformly in $Z\in \Xi_a$. Then the DLR consistency of the 
specification implies
$$
\int_\Xi \expect_{\rho^\sharp_T(\cdot|\YY)}[F_S|\YY] d\mu^\sharp(\YY) = 
\lim_{T_2\to+\infty} \expect_{\rho^\sharp_{T_2}(\cdot|\ZZ)}[F_S] 
= E_{\mu^\sharp}[F_S],
$$
proving the statement.
\end{proof}

\subsection{Typical path configurations}
In this section we show that most of $\mu$'s weight is
concentrated on paths that can be characterized by a growth 
condition. The proof of this depends on a lemma which was
already shown in \cite{LM}, however, we include it here for 
making the presentation more self-contained.   
\begin{theorem}
Suppose (\ref{iuccond}) holds with $a = b = s$, and
$\mu$ is a probability measure obtained by Theorem \ref{exist} 
for $V$ and $W$. Then there is a number $C > 0$ and a functional 
$R(X)$, such that 
\begin{equation}
|X_t| \leq \left(C \log (|t| + 1)\right)^{1/(s+1)} + R(X)
\end{equation}
$\mu$-almost surely.
\label{tpb}
\end{theorem}
\begin{proof}
The strategy of proving this theorem is to derive the typical 
behaviour of $\mu$ from the typical behaviour of the reference 
process. This follows through Lemma \ref{ldphi} below. Then
combining this lemma with Corollary~\ref{cor:equiv-mu-nu} gives
\begin{equation*}
  \begin{split}
\mu\left(\max_{0 \leq t \leq 1} |X_t| \geq a\right)
& \;\le \; 
C \left[\expect_\nu(1_{\max_{0 \leq t \leq 1} |X_t| \geq a})\right]^{1/2}      
\le  c'  e^{-\theta' a^{s+1}},
  \end{split}
\end{equation*}
with some constants $c',\theta' > 0$. Thus under the stationary 
measure $\mu$ 
\begin{equation}
\mu\left(\max_{n \leq t \leq n+1} |X_t| \; \geq \; 
(k\log n)^{1/(s+1)}\right) \; \leq \; \mbox{const} 
\frac{1}{n^{k \theta'}}
\end{equation}
holds. Choosing $k$ so that $k \theta' > 1$, the 
Borel-Cantelli Lemma implies that $\mu$-almost surely
\begin{equation}
|X(t)| \; \leq \; (k \log t)^{1/(s+1)}
\end{equation}
for $t \geq T^*$, with $T^* = T^*(X)$ sufficiently large. 
Writing $R(X) = \max_{|t| \leq T^*} |X(t)|$ completes the proof. 
\end{proof}

Finally we prove the lemma used above. 
\begin{lemma}
Let $\nu$ be the measure of the It\^o-process for $V$ satisfying 
(\ref{iuccond}) with exponent $s > 2$, and $a > 0$. Then 
there exist $C > 0$ and $\theta > 0$ such that
\begin{equation}
\nu\left(\max_{0 \leq t \leq 1} |X_t| \geq a \right) 
\; \leq C \; e^{-\theta a^{s+1}}.
\end{equation}
\label{ldphi}
\end{lemma}
\begin{proof}
For the underlying It\^o-process we have the Dirichlet operator 
on $L^2(\R^d,d\nu)$
\begin{equation}
Lf = -\Delta f + 2 (\nabla \log \Psi, \nabla f)
\end{equation}
and Dirichlet form 
\begin{equation}
{\cal E}(f,f) = - \int f \Delta f d\omega + 
2 \int f (\nabla \log \Psi, \nabla f) d\omega,
\end{equation}
with $d\omega = \Psi^2 dx$, as before. By using Varadhan's Lemma 
(see Lemma 1.12, \cite{KV86}), for any $f \in L^2(d\omega)$ and 
every $N > 0$  
\begin{equation}
\nu \left(\max_{0\leq t\leq 1} |f(X_t)| \geq N\right) 
\; \leq \; \frac{3}{N} \sqrt{{\cal E}(f,f) + (f,f)}
\label{kv}
\end{equation}
holds. Choose $f = f_a := 1_{\{x \in \R^d: |x| \geq a\}} * \phi$ 
by picking a mollifier $\phi$ (with $||\phi||_\infty < \infty$) 
so that the above convolution is in the domain of $L$. This can be 
chosen so that the smoothing of the edges of the indicator function 
takes place in a sphere $S(a)$ of radius $a$ centred at the origin, 
i.e., with a suitable $\e > 0$ we take $f_a(x) = 1$ for $x \in \R^d 
\smallsetminus S(a+\e)$, $f_a(x) = 0$ for $x \in S(a-\e)$, and $f_a$ 
is a sufficiently smooth function ${\tilde f}_a$ otherwise. Denote 
these three domains by $D_1$, $D_2$ and $D_3$, respectively. Setting 
$N=1$ in (\ref{kv}) yields
$$
\nu \left(\max_{0 \leq t \leq 1} |X_t| \geq a \right) \leq \; 
3 \sqrt{||f_a||^2_{L^2(d\omega)} + (f_a, L f_a)_{L^2(d\omega)}}.
$$
Moreover, we have
$$
||f_a||^2_{L^2(d\omega)} = \int f_a^2 (x) d\omega(x) = 
\int_{D_1} d\omega(x) + \int_{D_3} {\tilde f}_a^2(x) d\omega(x).
$$
Under the hypothesis on $V$ the standard estimate 
$
\Psi(x) \le C e^{-\theta |x|^{s+1}}
$
holds by Carmona's results \cite{Ca78} for the ground state $\Psi$, 
with some $C,\theta > 0$. This bound further leads to
\begin{equation}
\int_{D_1}d\omega(x) \; \leq \; c e^{-\theta a^{s+1}},
\end{equation}
where $c,\theta > 0$ are independent of $a$. A similarly estimate is 
valid for $D_3$. On the other hand, since ${\tilde f}_a$ is smooth 
enough and $\max_{D_3} \{ |\nabla {\tilde f}_a|,|\Delta {\tilde f}_a|, 
\Delta {\tilde f}_a^2 \} \leq m < \infty$, we get
$$
({\tilde f}_a, L{\tilde f}_a) \leq c' e^{-\theta a^{s+1}},
$$
with suitable $c'> 0$. A similar estimate is obtained also for the 
remaining two domains. 
\end{proof}

\subsection{Mixing properties}
Since $\mu$ is constructed in a way that offers no immediate 
access to computations with this measure, it is important to 
derive further basic information on $\mu$ by using the cluster
expansion. We give here one last result of this paper. 
\begin{theorem}
Let $F,G$ be two bounded functions, the first measurable with respect 
to $\F_I$, the second with respect to $\F_J$, where $I,J$ are distinct 
intervals of the partition considered in the cluster expansion above. 
Then the estimate on the covariance 
\begin{eqnarray*}
&& \cov_\mu (F;G) = \mathbb{E}_\mu[F G] - \mathbb{E}_\mu[F] \; 
\mathbb{E}_\mu[G] \\
&& |\cov_\mu \; (F;G)| \;\leq \; \mbox{\rm{const}} \; 
\frac{\sup{|F|} \sup{|G|}} {|t-s|^\vartheta + 1}
\end{eqnarray*}
holds, where $\vartheta > 0$ and the constant prefactor is 
independent of $F,G$.
\end{theorem}

\medskip
\begin{proof}
First recall formula~(\ref{eq:limit-eq-cluster}) which applied to $F$ 
(and similarly to $G$) gives
\begin{equation}
\mathbb{E}_\mu [F] = \mathbb{E}_\nu [F] f^{I} + 
\sum_{\Ga_0: I \cap \Ga_0^* \neq \emptyset} K^{\Ga_0}(F) f^{\Ga_0^*},
\label{F}
\end{equation}
where we let $K_{\Ga_0}(F) = \mathbb{E}_{\chi} [F \kappa_{\Ga_0}] $. 
Furthermore, consider $f^A$ estimated as before like $|f^A| \leq 
2^{|A|}$. For $A_1 \cap A_2 = \emptyset$ we have
\begin{equation}
|f^{A_1 \cup A_2} - f^{A_1} f^{A_2}| \; \leq \; \mbox{const} \; 
\frac{2^{|A_1| + |A_2|}}{\dist(A_1,A_2)^\zeta},
\label{cost}
\end{equation}
with some $\zeta > 0$. This estimate can easily be obtained by the 
general results in \cite{MM}. Now we  write
\begin{equation}
\begin{split}
\mathbb{E}_\mu [F G] 
& =  
\mathbb{E}_\nu [F] \; \mathbb{E}_\nu [G] f^{I \cup J} \\
& \qquad + 
\sum_{\Ga_1 \atop  \Ga^*_1 \cap I \neq \emptyset , \Ga^*_1 \cap J 
= \emptyset} \mathbb{E}_{\chi} [F \kappa_{\Ga_1}] \; 
\mathbb{E}_\chi [G]  f^{J \cup \Ga^*_1}
+ \sum_{\Ga_2 \atop  \Ga^*_2 \cap J \neq \emptyset , \Ga^*_2 \cap I 
= \emptyset} \mathbb{E}_{\chi} [G \kappa_{\Ga_2}] \; 
\mathbb{E}_\chi [F]  f^{I \cup \Ga^*_2}
\\
& \qquad + 
\sum_{\Ga_1,\Ga_2 : \Ga^*_1 \cap \Ga^*_2 = \emptyset \atop  \Ga^*_2 
\cap J \neq \emptyset , \Ga^*_2 \cap I \neq \emptyset}  
\mathbb{E}_\chi [F \kappa_{\Ga_1}] \; 
\mathbb{E}_{\chi} [G \kappa_{\Ga_2}]  f^{\Ga^*_1 \cup \Ga^*_2}
\sum_{\Ga \atop  \Ga^* \cap J \neq \emptyset , \Ga^* \cap I \neq 
\emptyset}  \mathbb{E}_\chi [F G \kappa_{\Ga}] f^{\Ga^*}.
\end{split}
\end{equation}
From here and (\ref{F}) we obtain
\begin{equation}
\begin{split}
\cov_\mu(F; G) 
& = 
\mathbb{E}_\chi [F] \; \mathbb{E}_\chi [G](f^{I\cup J}-f^{I}f^{J}) 
\\ & \qquad + 
\sum_{\Ga_1 \atop  \Ga^*_1 \cap I \neq \emptyset , \Ga^*_1 \cap J 
= \emptyset} \mathbb{E}_{\chi} [F \kappa_{\Ga_1}] \; 
\mathbb{E}_\chi [G]  (f^{J \cup \Ga^*_1}- f^{J} f^{\Ga^*_1})
 \\ & \qquad 
+ \sum_{\Ga_2 \atop  \Ga^*_2 \cap J \neq \emptyset , \Ga^*_2 \cap I 
= \emptyset} \mathbb{E}_{\chi} [G \kappa_{\Ga_2}] \; 
\mathbb{E}_\chi [F]  (f^{I \cup \Ga^*_2}-f^{I}f^{\Ga^*_2})
\\
& \qquad + 
\sum_{\Ga_1,\Ga_2 : \Ga^*_1 \cap \Ga^*_2 = \emptyset \atop  \Ga^*_2 
\cap J \neq \emptyset , \Ga^*_2 \cap I \neq \emptyset}  
\mathbb{E}_\chi [F \kappa_{\Ga_1}] \; \mathbb{E}_{\chi} 
[G \kappa_{\Ga_2}]  (f^{\Ga^*_1 \cup \Ga^*_2}-f^{\Ga^*_1}f^{\Ga^*_2})
 \\ & \qquad
+ 
\sum_{\Ga \atop  \Ga^* \cap J \neq \emptyset , \Ga^* \cap I \neq 
\emptyset}  \mathbb{E}_\chi [F G \kappa_{\Ga}] f^{\Ga^*}
 \\ &\qquad
- 
\sum_{\Ga \atop  \Ga^* \cap J \neq \emptyset , \Ga^* \cap I \neq 
\emptyset}  \mathbb{E}_\chi [F ] \mathbb{E}_\chi [ G \kappa_{\Ga}] 
f^{\Ga^*} f^I f^{\Ga^*}
 \\ & \qquad
- 
\sum_{\Ga \atop  \Ga^* \cap J \neq \emptyset , \Ga^* \cap I \neq 
\emptyset}  \mathbb{E}_\chi [F \kappa_{\Ga}] \mathbb{E}_\chi [ G ] 
f^{\Ga^*} f^J f^{\Ga^*}
 \\ & \qquad 
- \sum_{\Ga_1,\Ga_2 : \Ga^*_1 \cap \Ga^*_2 \neq \emptyset \atop  
\Ga^*_2 \cap J \neq \emptyset , \Ga^*_2 \cap I \neq \emptyset}  
\mathbb{E}_\chi [F \kappa_{\Ga_1}] \; \mathbb{E}_{\chi} 
[G \kappa_{\Ga_2}]  f^{\Ga^*_1}f^{\Ga^*_2}
\end{split}
\end{equation}
For estimating the first four terms at the right hand side above we 
use (\ref{cost}) along with the bound
\begin{equation}
|\mathbb{E}_{\p} [F \kappa_{\Ga}]| \; \leq \; \frac{\sup |F|}
{(\diam \Ga^*)^{\zeta'} + 1} E_{\Ga}(\delta',\eps')
\end{equation}
$i = 1, 2$, where $E_{\Ga}(\delta',\eps')$ is the function 
appearing at the right hand side of estimate~(\ref{clustest}) with 
slightly modified entries ($\delta',\eps'$ instead of $\delta,\eps$; 
$\delta' > 1$) so that (\ref{convv}) still holds. Here $\zeta' = 
\delta - \delta' > 0$, and we used in addition that 
\begin{equation}
\frac{1}{(\diam \Ga^*)^{\zeta'} + 1} \; 
\frac{1}{\dist(I,\Ga^*)^\zeta + 1} \; \leq \; 
\frac{1}{\dist(I,J)^\vartheta + 1}
\end{equation}
whenever $J \cap \Ga^* \neq \emptyset$ (similarly for $I$), and
\begin{equation}
\frac{1}{(\diam \Ga^*_1)^{\zeta'} + 1} \; 
\frac{1}{(\diam \Ga^*_2)^{\zeta'} + 1} \; 
\frac{1}{\dist(\Ga^*_1,\Ga^*_2)^\zeta} \; \leq \; 
\frac{1}{\dist(I,J)^\vartheta + 1}
\end{equation}
for $I \cap \Ga^*_1 \neq \emptyset$, $J \cap \Ga^*_2 \neq \emptyset$, 
and $\vartheta = \min \{\zeta,\zeta'\} > 0$.

Next, in the fifth term above we use that $\diam \Ga^* \geq \dist(I,J)$ 
whenever $I \cap \Ga^* \neq \emptyset$, $J \cap \Ga^* \neq \emptyset$, 
and that
\begin{equation}
|\mathbb{E}_{\chi} [F G \kappa_\Ga]| \; \leq \; \sup|F| \sup|G| 
\frac{E_\Ga(\delta',\eps')} {(\diam \Ga^*)^{\zeta'} + 1}.
\end{equation}
For the remaining three terms in the sum above we apply the same argument. 
Thus for the full sum the corresponding bounds become
\begin{equation}
\mbox{const} \; \frac{\sup|F| \sup|G|}{\dist(I,J)^\vartheta + 1} 
E_{\Ga_1}(\delta',\eps') 
E_{\Ga_2}(\delta',\eps') \; 2^{|\Ga_1^*| + |\Ga^*_2|}, 
\end{equation}
whenever $I \cap \Ga^*_1 \neq \emptyset$, $J \cap \Ga^*_2 \neq \emptyset$, 
respectively 
\begin{equation}
\mbox{const} \; \frac{\sup|F| \sup|G|}{\dist(I,J)^\vartheta + 1} 
E_{\Ga}(\delta',\eps') \; 2^{|\Ga^*|} 
\end{equation}
in the other cases. Then using  Proposition~\ref{p2} we can prove 
boundedness of the sums over  $\Ga_1, \Ga_2$ or $\Ga$, concluding the 
proof.
\end{proof}

\end{document}